\documentclass{article} 
\usepackage[utf8]{inputenc}

\usepackage[english]{babel} 
\usepackage{amsmath} 
\usepackage{amsfonts} 
\usepackage{amssymb} 
\usepackage{amsthm}
\usepackage{mathdots}    
\usepackage{url} 
\usepackage{semantic}
\usepackage{cmll}
\usepackage{subfigure}
\usepackage{enumerate}   
\usepackage{chngcntr}     
\usepackage{bussproofs}
\usepackage{tikz}
\usepackage[final, hyperindex, bookmarks]{hyperref}

\newtheorem{theorem}{Theorem} [section]

\newtheorem{lemma}[theorem]{Lemma} 
\newtheorem{corollary}[theorem]{Corollary} 
\newtheorem{proposition}[theorem]{Proposition} 

\theoremstyle{definition}
\newtheorem{example}[theorem]{Example} 
\newtheorem{definition}[theorem]{Definition} 

\theoremstyle{remark}
\newtheorem{remark}[theorem]{Remark} 

\counterwithin{equation}{section}
\counterwithin{figure}{section}
\counterwithin{table}{section}

\newcommand{\sub}[2]{#1[#2]}
\newcommand{\arrow}{\rightarrow}        
\newcommand{\trarrow}{\stackrel{*}{\rightarrow}}        
\newcommand{\limp}{\multimap} 
\newcommand{\bang}{\oc} 

\newcommand{\hyp}[3]{#1:(#2, #3)}
\newcommand{\letm}[3]{{\sf let} \ ! #1 = #2 \ {\sf in} \ #3}    
\newcommand{\tertype}{{\bf 1}}
\newcommand{\behtype}{{\bf B}}
\newcommand{\pair}[2]{\langle #1 , #2 \rangle} 
\newcommand{\s}[1]{{\sf #1}}    
\newcommand{\w}[1]{{\it #1}}    
\newcommand{\qqs}[2]{\forall\, #1\;\: #2} 
\newcommand{\vc}[1]{{\bf #1}} 
\newcommand{\set}[1]{\{#1\}}
\newcommand{\st}[2]{{\sf set}(#1,#2)}

\newcommand{\rgtype}[2]{{\it {\sf Reg}_{#1} #2}}
\newcommand{\get}[1]{{\sf get}(#1)}

\newcommand{\store}[2]{(#1 \leftarrow #2)}
\newcommand{\regtype}[2]{{\sf Reg}_{#1} #2}

\newcommand{\letb}[3]{\mathsf{let}\;\oc #1 = #2\;\mathsf{in}\;#3}

\newcommand{\lambdab}[2]{\lambda^{\oc} #1.#2}
\newcommand{\ftrad}[1]{(#1)^{-}}
\newcommand{\fst}[1]{\textsf{fst}\;#1}
\newcommand{\snd}[1]{\textsf{snd}\;#1}
\newcommand{\lambdal}[2]{\lambda #1^{#2}}
\newcommand{\coerc}[1]{\mathcal{C}(#1)}

\newcommand{\nat}{\mathsf{N}}

\newcommand{\reduc}{\rightarrow}
\newcommand{\treduc}{\stackrel{*}{\rightarrow}}

\newcommand{\para}{\mid}
\newcommand{\churchn}[1]{\overline{#1}}
\newcommand{\letbang}{\mathsf{let}\,\oc}
\newcommand{\occ}{\omega}

\newcommand{\listtype}[1]{\mathsf{List}\:#1}

\newcommand{\lit}{\mathsf{list\_it}}
\newcommand{\churchl}[1]{[#1]}
\newcommand{\add}{\mathsf{add}}
\newcommand{\mult}{\mathsf{mult}}
\newcommand{\iit}{\mathsf{int\_it}}
\newcommand{\git}{\mathsf{int\_git}}
\newcommand{\size}[1]{|#1|}

\newcommand{\zero}{\mathsf{zero}}
\newcommand{\succe}{\mathsf{succ}}

\newcommand{\lambdaf}{\lambda^{\oc}}
\newcommand{\lambdar}{\lambda^{\oc\mathsf{R}}}   

\newcommand{\lambdafd}{\lambda^{\oc}_{\delta}}
\newcommand{\lambdard}{\lambda^{\oc\mathsf{R}}_{\delta}}   

\newcommand{\lambdaft}{\lambda^{\oc}_{EA}}
\newcommand{\lambdart}{\lambda^{\oc\mathsf{R}}_{EA}}   

\newcommand{\fv}[1]{\mathsf{FV}(#1)}
\newcommand{\fo}[2]{\mathsf{FO}(#1,#2)}
\newcommand{\ie}{\emph{i.e.}\;}

\begin{document} 
 
\title{An Elementary Affine $\lambda$-calculus\\with Multithreading and
  Side Effects\thanks{Work partially supported by project
  ANR-08-BLANC-0211-01 ``COMPLICE'' and the Future and Emerging Technologies (FET)
programme within the Seventh Framework Programme for Research of the
European Commission, under FET-Open grant number: 243881 (project
CerCo).}}

\author{\textsc{Antoine Madet} \and \textsc{Roberto M. Amadio} \and \\
    Laboratoire PPS, Universit\'e Paris Diderot\and
  {\small \url{{madet,amadio}@pps.jussieu.fr}}}

\date{}
\maketitle 

\vspace{7cm}
\begin{abstract}
  Linear logic provides a framework to control the complexity of
higher-order functional programs. We present an extension of this
framework to programs with multithreading and side effects focusing on
the case of elementary time.  Our main contributions are as follows.
First, we provide a new combinatorial proof of termination in
elementary time for the functional case.  Second, we develop an
extension of the approach to a call-by-value $\lambda$-calculus with
multithreading and side effects. Third, we introduce an elementary affine type system
  that guarantees the standard subject reduction and progress
  properties. Finally, we illustrate the
programming of iterative functions with side effects in the
presented formalism.
\end{abstract}

\newpage
\tableofcontents
\newpage

\section{Introduction}
There is a well explored framework based on Linear Logic to control
the complexity of higher-order functional programs. In particular,
\emph{light logics}~\cite{LLL,Danos,Asperti02} have led to a
polynomial light affine
$\lambda$-calculus~\cite{DBLP:journals/aml/Terui07} and to various
type systems for the standard $\lambda$-calculus guaranteeing that a
well-typed term has a bounded
complexity~\cite{CoppolaMartini2006,LMCS2008,Baillot05}.  Recently,
this framework has been extended to a higher-order
process calculus~\cite{LagoExpress} and a functional language with
recursive definitions~\cite{BaillotGM10}. In another direction,
the notion of \emph{stratified region}~\cite{Boudol10,Amadio09} has
been used to prove the termination of higher-order multithreaded
programs with side effects.

Our general goal is to extend the framework of light logics to a
higher-order functional language with multithreading and side effects
by focusing on the case of elementary
time~\cite{Danos}. The key point is that termination does not rely
anymore on stratification but on the notion of depth which is standard in light
logics.  Indeed, light logics suggest that complexity can be tamed through a
fine analysis of the way the depth of the occurrences of a
$\lambda$-term can vary during reduction.

Our core functional calculus is a $\lambda$-calculus extended
with a constructor $`\oc$' (the modal operator of linear logic)
marking duplicable terms and a related $\letbang$ destructor.  The
depth of an occurrence in a $\lambda$-term is the number of $\oc's$
that must be crossed to reach the occurrence. Our contribution
can be described as follows.

\begin{enumerate}
\item In Section~\ref{well-formed-pure-sec} we propose a formal system
  called {\em depth system} that controls the depth of the occurrences
  and which is a variant of a system proposed
  in~\cite{DBLP:journals/aml/Terui07}. We show that terms
  well-formed in the depth system are guaranteed to terminate in
  elementary time under an arbitrary reduction strategy.  The proof is
  based on an original combinatorial analysis of the depth system
  (\cite{Danos} assumes a specific reduction strategy
  while~\cite{DBLP:journals/aml/Terui07} relies on a standardization
  theorem).

\item In Section~\ref{modal-sec}, following previous work on an
  affine-intuitionistic system~\cite{air-hal}, we extend the
  functional core with parallel composition and operations producing
  side effects on an `abstract' notion of state. We analyse the impact
  of side-effects operations on the depth of the occurrences and
  deduce an extended depth system. We show
  that it still guarantees termination of programs in elementary time under a
  natural call-by-value evaluation strategy.

\item In Section~\ref{type-system-sec}, we refine the depth system
  with a second order (polymorphic) elementary affine type system and
  show that the resulting system enjoys subject reduction and progress
  (besides termination in elementary time).

\item Finally, in
Section~\ref{sec-expressivity}, we discuss the expressivity of the
resulting type system. On the one hand we check that the usual
encoding of elementary functions goes through.  On the other hand, and
more interestingly, we provide examples of iterative (multithreaded)
programs with side effects.
\end{enumerate}

  The $\lambda$-calculi introduced are
summarized in Table \ref{summary-table}.  For each concurrent language
there is a corresponding functional fragment and each language
(functional or concurrent) refines the one on its left hand side.  The
elementary complexity bounds are obtained for the $\lambdafd$ and
$\lambdard$ calculi while the progress property and the expressivity
results refer to their typed refinements $\lambdaft$ and $\lambdart$,
respectively.  Proofs are available in Appendix \ref{proofs-sec}.
\begin{table}[h]
\[
\begin{array}{|c|lllll|}

\hline
\ \mbox{Functional} & \lambdaf &\supset &\lambdafd &\supset &\lambdaft   \ \\
      \cap      &&&&& \\
\ \mbox{Concurrent} &\lambdar  &\supset &\lambdard &\supset &\lambdart   \ \\\hline

\end{array}
\]
\caption{Overview of the $\lambda$-calculi considered}\label{summary-table}
\end{table}


\section{Elementary Time in a Modal $\lambda$-calculus}
\label{well-formed-pure-sec}
In this section, we present our core functional calculus, a related
depth system, and show that every term which is well-formed in the
depth system terminates in elementary time under an arbitrary reduction
strategy.

\subsection{A Modal $\lambda$-calculus}
We introduce a modal $\lambda$-calculus called $\lambdaf$. It is very
close to the \emph{light affine $\lambda$-calculus} of
Terui~\cite{DBLP:journals/aml/Terui07} where the paragraph modality
`$\mathsection$' used for polynomial time is dropped and where the
`$\oc$' modality is relaxed as in elementary linear
logic~\cite{Danos}.

\subsubsection{Syntax}
Terms are described by the grammar in Table~\ref{grammar}:
\begin{table}[h]
$$
\begin{array}{c}
M, N ::= x, y, z\ldots \mid \lambda x.M \mid MN \mid \bang M \mid
\letb{x}{N}{M}
\end{array}
$$
\caption{Syntax of $\lambdaf$}
\label{grammar}
\end{table}
We find the usual set of variables, $\lambda$-abstraction and
application, plus a modal operator `$\oc$' (read \emph{bang}) and a
$\letbang$ operator. 
We define $\bang^0 M = M$ and $\bang^{n+1} M = \bang (\bang^{n}M)$. In
the terms $\lambda x.M$ and $\letb{x}{N}{M}$ the occurrences of $x$ in
$M$ are bound. The set of free variables of $M$ is denoted by
$\fv{M}$. The number of free occurrences of $x$ in $M$ is denoted by
$\fo{x}{M}$. $\sub{M}{N/x}$ denotes the term $M$ in which each free
occurrence of $x$ has been substituted by the term $N$.

Each term has an \emph{abstract syntax tree} 
as exemplified in Figure~2.\ref{fsyntax-tree}.
\begin{figure}[h]
  \centering
    \subfigure[]{\label{fsyntax-tree}
    \begin{tikzpicture}[font=\normalsize,level distance=0.75cm]
      \node {$\lambda x$}
      child {node {$\letbang y$}
        child {node {$x$}}
        child {node {$\oc$}
          child {node {$@$}
            child {node {$y$}}
            child {node {$y$}}
          }
        }
      };
    \end{tikzpicture}}
  \qquad
    \subfigure[]{\label{faddress-tree}
    \begin{tikzpicture}[font=\normalsize,level distance=0.75cm]
      \node {$\epsilon$}
      child {node {$0$}
        child {node {$00$}}
        child {node {$01$}
          child {node {$010$}
            child {node {$0100$}}
            child {node {$0101$}}
          }
        }
      };
    \end{tikzpicture}}
  \qquad
    \subfigure[]{\label{fdepth-tree}
    \begin{tikzpicture}[font=\normalsize,level distance=0.75cm]
      \node {$0$}
      child {node {$0$}
        child {node {$0$}}
        child {node {$0$}
          child {node {$1$}
            child {node {$1$}}
            child {node {$1$}}
          }
        }
      };
    \end{tikzpicture}}
  \caption{Syntax tree of the term $\lambda x.\letb{y}{x}{\oc
      (yy)}$, addresses and depths}
  \label{ftrees}
\end{figure}
A path starting from the root to a node of the tree denotes an
\emph{occurrence} of the program that is denoted by
 a word $w \in \set{0,1}^{*}$ (see Figure~2.\ref{faddress-tree}).

We define the notion of $\emph{depth}$:
\begin{definition}[depth]\label{depth-functional}
  The \emph{depth} $d(w)$ of an occurrence $w$ is the number of $\oc$'s
  that the path leading to $w$ crosses.  The depth $d(M)$ of
  a term $M$ is the maximum depth of its occurrences.
\end{definition}
In Figure~2.\ref{fdepth-tree}, each occurrence is labelled with its
depth. Thus $d(\lambda x.\letb{y}{x}{\oc (yy)})=1$.  In
particular, the occurrence $01$ is at depth $0$; what
matters in computing the depth of an occurrence is the number of $\oc$
that precedes strictly the occurrence.

\subsubsection{Operational Semantics}
We consider an arbitrary reduction strategy. Hence, an evaluation
context $E$ can be any term with exactly one occurrence of a special
variable $[~]$, the `hole'. $E[M]$ denotes $E$ where the hole has been
substituted by $M$. The reduction rules are given in
Table~\ref{semantics-pure}.
\begin{table}[h]
  $$
  \begin{array}{r c l}
    E[(\lambda x.M)N] &\arrow& E[\sub{M}{N/x}]\\
    E[\letm{x}{\oc N}{M}] &\arrow& E[\sub{M}{N/x}]
  \end{array}
  $$
  \caption{Operational semantics of $\lambdaf$}
\label{semantics-pure}
\end{table}
The $\letbang$ is `filtering' modal terms and `destructs' the bang of
the term $\oc
N$  after substitution. In the sequel, $\treduc$ denotes the
reflexive and transitive closure of $\reduc$.

\subsection{Depth System}
\label{subsec-depth-sys}
By considering that deeper occurrences have less weight than shallow
ones, the proof of termination in elementary time~\cite{Danos} relies
on the observation that when reducing a redex at depth $i$ the
following holds:
\begin{enumerate}[\quad (1)]
\item \label{fact1} the depth of the term does not increase,
\item \label{fact2} the number of occurrences at depth $j < i$ does not increase,
\item \label{fact3} the number of occurrences at depth $i$ strictly decreases,
\item \label{fact4} the number of occurrences at depth $j > i$ may be
  increased by a multiplicative factor $k$ bounded by the number of
  occurrences at depth $i+1$.
\end{enumerate}

Theses properties can be guaranteed by the following requirements:
\begin{enumerate}[(i)]
\item in $\lambda x.M$, $x$ may occur at
most once in $M$ and at depth $0$, 
\item in $\letb{x}{M}{N}$, $x$ may
occur arbitrarily many times in $N$ and at depth $1$.
\end{enumerate}

Hence, the rest of this section is devoted to the introduction of a
set of inferences rules called depth system. Every term which is valid
in the depth system will terminate in elementary time. First, we
introduce the judgement:
$$
\Gamma \vdash^\delta M
$$
where $\delta$ is a natural number and the context $\Gamma$ is of the
form $x_1:\delta_1,\ldots,x_n:\delta_n$.  We write $dom(\Gamma)$ for
the set $\set{x_1,\ldots,x_n}$. It should be interpreted as
follows:
\begin{quotation}
  The free variables of $\oc^{\delta} M$ may only occur at the depth
  specified by the context $\Gamma$.
\end{quotation}

The inference rules of the depth system are presented in Table~\ref{depth-system}.
\begin{table}[h]
\[
\begin{array}{c}

\inference
{}
{\Gamma,x:\delta \vdash^\delta x} \\ \\

\inference
{
\Gamma,x:\delta \vdash^\delta M & \fo{x}{M} \leq 1}
{\Gamma \vdash^\delta \lambda x.M}

\qquad 
\inference{\Gamma\vdash^\delta M & \Gamma\vdash^\delta N}
{\Gamma \vdash^\delta MN} \\ \\

\inference
{\Gamma \vdash^\delta N & \Gamma,x:(\delta+1) \vdash^\delta M}
{\Gamma \vdash^\delta \letm{x}{N}{M}}

\qquad

\inference
{\Gamma \vdash^{\delta+1} M}
{\Gamma \vdash^\delta \bang M}

\end{array}
\]
\caption{Depth system: $\lambdafd$}
\label{depth-system}
\end{table}

We comment on the rules.
The variable rule says that the current depth
of a free variable is specified by the context. The $\lambda$-abstraction rule
requires that the occurrence of $x$ in $M$ is at the same depth as the
formal parameter; moreover it occurs at most once so  that no duplication
is possible at the current depth (Property~\eqref{fact3}). The application
rule says that we may only apply two terms if they are at the same
depth. The $\letbang$ rule requires that the bound occurrences of $x$
are one level deeper than the current depth; note that there is no restriction
on the number of occurrences of $x$ since duplication would
happen one level deeper than the current depth. Finally, the bang rule is
better explained in a bottom-up way: crossing a modal occurrence
increases the current depth by one.

\begin{definition}[well-formedness]
\label{def-well-formed}
A term $M$ is \emph{well-formed} if for some $\Gamma$ and $\delta$ a
judgement $\Gamma \vdash^\delta M$ can be derived.
\end{definition}

\begin{example}
The term of Figure~\ref{fsyntax-tree} is well-formed according to our depth system:
\begin{prooftree}
  \AxiomC{$x:\delta \vdash^{\delta} x$}
  \AxiomC{$x:\delta, y:\delta+1 \vdash^{\delta + 1} y$}
  \AxiomC{$x:\delta, y:\delta+1 \vdash^{\delta + 1} y$}
  \BinaryInfC{$x:\delta, y:\delta+1 \vdash^{\delta + 1} yy$}
  \UnaryInfC{$x:\delta, y:\delta+1 \vdash^\delta \oc (yy)$}
  \BinaryInfC{$x:\delta \vdash^\delta \letb{y}{x}{\oc (yy)}$}
  \UnaryInfC{$\vdash^\delta \lambda x.\letb{y}{x}{\oc (yy)}$}
\end{prooftree}  
On the other hand, the followings term is not valid:
$$P = \lambda x.\letb{y}{x}{\oc(y\oc(yz))}$$
Indeed, the second occurrence of $y$ in $\oc (y \oc
(yz))$ is too deep of one level, hence reduction may increase the
depth by one. For example, $P \oc\oc N$ of depth $2$ reduces to $\oc(\oc N \oc(\oc
N)z)$ of depth $3$.
\end{example}

\begin{proposition}[properties on the depth system]
\label{prop-depth}
The depth system satisfies the following properties:
\begin{enumerate}
\item If $\Gamma \vdash^\delta M$ and $x$ occurs free in $M$ then 
$x:\delta'$ belongs to $\Gamma$ and all occurrences of $x$ in 
$\bang^\delta M$ are at depth $\delta'$.

\item If $\Gamma \vdash^\delta M$ then $\Gamma,\Gamma'\vdash^\delta M$.

\item If $\Gamma,x:\delta' \vdash^\delta M$ and $\Gamma \vdash^{\delta'} N$ 
then $d(\bang^\delta \sub{M}{N/x}) \leq \mathit{max}(d(\bang^\delta M),d(\bang^{\delta'} N))$ and $\Gamma \vdash^\delta \sub{M}{N/x}$.

\item If $\Gamma \vdash^0 M$ and $M\arrow N$ then $\Gamma \vdash^0 N$ and 
$d(M) \geq d(N)$.

\end{enumerate}
\end{proposition}

\subsection{Elementary Bound}\label{termination-pure-sec}
In this section, we prove that well-formed terms terminate in
elementary time under an arbitrary reduction strategy.  To this end,
we define a measure on terms based on the number of occurrences at
each depth.

\begin{definition}[measure]
\label{def-measure}
Given a term $M$  and $0\leq i \leq d(M)$, 
let $\occ_i(M)$ be the number of occurrences in $M$ of
depth $i$ increased by $2$ (so $\occ_i(M)\geq 2$).
We define $\mu_n^i(M)$ for  $n \geq i \geq  0$ as follows:
\[
\mu_n^i(M)=(\occ_{n}(M),\ldots,\occ_{i+1}(M),\occ_{i}(M) )
\] 
We write $\mu_n(M)$ for $\mu_n^0(M)$.
We order the vectors of $n+1$ natural number with the (well-founded) 
lexicographic order $>$ from right to left.
\end{definition}

We derive a termination property by observing that the
measure strictly decreases during reduction. 
\begin{proposition}[termination]
\label{prop-termination}
If $M$ is well-formed, $M \arrow M'$ and $n\geq d(M)$ then $\mu_n(M)>\mu_n(M')$.
\end{proposition}
\begin{proof} 
  We do this by case analysis on the
  reduction rules:
\begin{itemize}
\item $M=E[(\lambda x.M_1)M_2] \reduc M' = E[\sub{M_1}{M_2/x}]$\\
Let the occurrence of the redex $(\lambda x.M_1)M_2$ be
at depth $i$.
The restrictions on the formation of terms require that
$x$ occurs at most once in $M_1$ at depth $0$.
Then $\occ_i(M) -3 \geq \occ_i(M')$ because we remove
the nodes for application and $\lambda$-abstraction and either 
$M_2$ disappears or the occurrence of the variable $x$ in $M_1$ 
disappears (both being at the same depth as the redex).
Clearly $\occ_j(M) = \occ_j(M')$ if $j \neq i$, hence 
\begin{equation}
\label{lambda-decrease}
\mu_n(M') \leq \\(\occ_n(M), \ldots, \occ_{i+1}(M),\occ_i(M) - 3,\mu_{i-1}(M))
\end{equation}
and
$\mu_n(M)>\mu_n(M')$.

\item $M=E[\letm{x}{!M_2}{M_1}] \reduc M' = E[\sub{M_1}{M_2/x}]$\\
Let the occurrence of the redex 
$\letm{x}{!M_2}{M_1}$ be at depth $i$.
The restrictions on the formation of terms require that
$x$ may only occur in $M_1$ at depth $1$ and hence in $M$
at depth $i+1$. 
We have that $\occ_i(M) = \occ_i(P) -2$ 
 because the $\letbang$ node disappear.
Clearly,  $\occ_j(M) = \occ_j(M')$ if $j<i$.
The number of occurrences of $x$ in $M_1$ is bounded by 
$k=\occ_{i+1}(M) \geq 2$. 
Thus if  $j>i$  then $\occ_{j}(M') \leq k \cdot
\occ_j(M)$. Let's write, for $ 0 \leq i \leq n$: 
$$\mu_n^i(M) \cdot k = (\occ_n(M) \cdot k,
\occ_{n-1}(M) \cdot k, \ldots, \occ_{i}(M) \cdot k)$$
Then we have
\begin{equation}
\label{na-letbang-decrease}
  \mu_n(M') \leq (\mu_n^{i+1}(M) \cdot k, \occ_i(M) -2, \mu_{i-1}(M))
\end{equation}
and finally 
$\mu_n(M)>\mu_n(M')$.
\end{itemize}
\end{proof}

We now want to show that termination is actually in elementary time.
We recall that a function $f$ on integers is elementary if there exists a $k$ such
that for any $n$, $f(n)$ can be computed in time $\mathcal{O}(t(n,k))$ where:
\begin{align*}
  t(n,0) = 2^n, \qquad 
  t(n,k+1) = 2^{t(n,k)}~.
\end{align*}
\begin{definition}[tower functions]\label{tower-def}
We define a family of tower functions\\\noindent
$t_\alpha(x_1,\ldots,x_n)$ by induction on $n$ 
where we assume $\alpha\geq 1$ and  $x_i\geq 2$:
\[
\begin{array}{rcl}
t_\alpha()    &=&0 \\
t_\alpha(x_1,x_2,\ldots,x_n) &=& (\alpha \cdot x_{1})^{2^{t_{\alpha}(x_{2},\ldots,x_{n})}} \quad n\geq 1
\end{array}
\]
\end{definition}

Then we need to prove the following crucial lemma.
\begin{lemma}[shift]\label{shift-lemma}
Assuming $\alpha \geq 1$ and $\beta \geq 2$, 
the following property holds for the tower functions 
with $x,\vc{x}$ ranging over numbers greater or equal to $2$:
$$t_\alpha(\beta \cdot x, x',\vc{x}) \leq t_\alpha(x,\beta \cdot x',\vc{x})$$
\end{lemma}

Now, by a closer look at the shape of the lexicographic ordering
during reduction, we are able to compose the decreasing measure with a
tower function.
\begin{theorem}[elementary bound]
\label{elementary-bound-pure}
Let $M$ be a well-formed term with $\alpha =d(M)$ and let 
$t_\alpha$ denote the tower function with $\alpha+1$ arguments.
If $M\arrow M'$ then $t_\alpha(\mu_\alpha(M)) > t_\alpha(\mu_\alpha(M'))$.
\end{theorem}
\begin{proof}
  We illustrate the proof for $\alpha = 2$ and the crucial case where $$M
  = \letb{x}{\oc M_1}{M_2} \arrow M' = \sub{M_1}{M_2/x}$$
  Let $\mu_2(M) = (x, y, z)$ such that $x = \occ_2(M)$, $y = \occ_1(M)$
  and $z = \occ_0(M)$. 
  We want to show that:
  $$
  t_2(\mu_2(M')) < t_2(\mu_2(M))
  $$
  We have:
  $$
  \begin{array}{rcll}
  t_2(\mu_2(M')) &\leq& t_2(x \cdot y, y \cdot y, z -2) & \text{by inequality~\eqref{na-letbang-decrease}}\\
  &\leq& t_2(x,y^3,z-2) & \text{by Lemma~\ref{shift-lemma}}
  \end{array}
  $$
  Hence we are left to show that:
  $$t_2(y^3,z-2) < t_2(y,z) \text{\quad\ie\quad} (2y^3)^{2^{2(z-2)}} <
  (2y)^{2^{2z}}$$
  We have:
  $$
  \begin{array}{rcll}
    (2y^3)^{2^{2(z-2)}} &\leq& (2y)^{3\cdot2^{2(z-2)}}
  \end{array}
  $$
  Thus we need to show:
  $$
  3\cdot2^{2(z-2)} < 2^{2z}
  $$
  Dividing by $2^{2z}$ we get:
  $$
  3 \cdot 2^{-4} < 1
  $$
  which is obviously true. Hence $  t_2(\mu_2(M')) < t_2(\mu_2(M))$.
\end{proof}

This shows that the number of reduction steps of a term $M$ is bound
by an elementary function where the height of the tower
depends on $d(M)$. We also note that if $M\trarrow M'$ then
$t_\alpha(\mu_\alpha(M))$ bounds the size of $M'$. Thus we can
conclude with the following corollary.
\begin{corollary}[elementary time normalisation]
The normalisation of terms of bounded depth can be performed
in time elementary in the size of the terms.
\end{corollary}

\section{Elementary Time in a Modal $\lambda$-calculus with Side
  Effects}
\label{modal-sec}
In this section, we extend our functional language with side effects
operations. By analysing the way side effects act on the depth of
occurrences, we extend our depth system to the obtained language. We
can then lift the proof of termination in elementary time to programs
with side effects that run with a call-by-value reduction strategy.

\subsection{A Modal $\lambda$-calculus with Multithreading and Regions}
We introduce a call-by-value modal $\lambda$-calculus
endowed with parallel composition and operations to read and write
\emph{regions}. We call it $\lambdar$.  A region is an {\em
  abstraction} of a set of dynamically generated values such as
imperative references or communication channels.  We regard
$\lambdar$ as an abstract, highly non-deterministic language which
entails complexity bounds for more concrete languages featuring
references or channels (we will give an example of such a language in Section
\ref{sec-expressivity}). To this end, it is enough to map the
dynamically generated values to their respective regions and observe
that the reductions in the concrete languages are simulated in
$\lambdar$ (see, {\em e.g.}, \cite{air-hal}). 

\subsubsection{Syntax}
The syntax of the language is described in Table~\ref{syntax}.
\begin{table}[h]
\begin{displaymath}
\begin{array}{rcll}
\multicolumn{3}{l}{x,y,\ldots}                                           &\mbox{(Variables)} \\
\multicolumn{3}{l}{r,r',\ldots}                                          &\mbox{(Regions)} \\
V&::=& * \mid r \mid x \mid \lambda x.M \mid \bang V                &\mbox{(Values)}\\
M&::=& V \mid MM \mid \bang M \mid \letm{x}{M}{M}            \\
&&  \st{r}{V} \mid \get{r} \mid (M\para M) &\mbox{(Terms)}\\ 
S&::=& \store{r}{V} \mid (S\para S)   &\mbox{(Stores)} \\
P&::=& M \mid S \mid (P\para P)           &\mbox{(Programs)}  \\
E&::=&[~]\mid EM \mid VE \mid \bang E \mid \letm{x}{E}{M}      &\mbox{(Evaluation Contexts)} \\
C&::=& [~] \mid (C\para P) \mid (P\para C)   &\mbox{(Static Contexts)}
\end{array}
\end{displaymath}
\caption{Syntax of programs: $\lambdar$}
\label{syntax}
\end{table}
We describe the new operators.
We have the usual set of variable $x,y,\ldots$ and a set of regions
$r,r',\ldots$. The set of values $V$ contains the unit constant $*$,
variables, regions, $\lambda$-abstraction and modal values $\oc V$
which are marked with the \emph{bang} operator `$\oc$'. The set of
terms $M$ contains values, application, modal terms $\oc M$, a
$\letbang$ operator, $\st{r}{V}$ to write the value $V$ at region $r$,
$\get{r}$ to fetch a value from region $r$ and $(M \para N)$ to
evaluate $M$ and $N$ in parallel. A store $S$ is the composition of
several stores $\store{r}{V}$ in parallel.  A program $P$ is a
combination of terms and stores.  Evaluation contexts follow a
call-by-value discipline. Static contexts $C$ are composed of parallel
compositions. Note that stores can only appear in a static context,
thus $M(M' \para \store{r}{V})$ is not a legal term.
We define $\oc^n( P \para P) = (\oc^n P \para \oc^n P)$, and  $\oc^n\store{r}{V} =
\store{r}{V}$. As usual, we
abbreviate $(\lambda z.N)M$ with $M;N$, where $z$ is not free in $N$.

Each program has an \emph{abstract syntax tree} 
as exemplified in Figure~3.\ref{esyntax-tree}.
\begin{figure}[h]
  \centering
    \subfigure[]{\label{esyntax-tree}
    \begin{tikzpicture}[font=\normalsize,level distance=0.75cm]
      \node {$\para$}[sibling distance=1.7cm]
      child {node {$\letbang x$}[sibling distance=1cm]
        child {node {$\get{r}$}}
        child {node {$\mathsf{set}(r)$}
          child {node {$\oc$}
            child {node {$x$}}}}}
      child {node {$r \leftarrow$}
        child {node {$ \oc$}
          child {node {$ \lambda x$}
            child {node {$@$}[sibling distance=1cm]
              child {node {$x$}}
              child {node {$*$}}}}}};
    \end{tikzpicture}}
  \qquad\qquad\qquad
    \subfigure[]{\label{eaddress-tree}
    \begin{tikzpicture}[font=\normalsize,level distance=0.75cm]
      \node {$\epsilon$}[sibling distance=1.7cm] 
      child {node {$0$}[sibling distance=1cm]
        child {node {$00$}}
        child {node {$01$}
          child {node {$010$}
            child {node {$0100$}}}}}
      child {node {$1$}
        child {node {$10$}
          child {node {$100$}
            child {node {$1000$}[sibling distance=1.3cm]
              child {node {$10000$}}
              child {node {$10001$}}}}}};
    \end{tikzpicture}}
\caption{Syntax tree  and addresses of $P = \letb{x}{\get{r}}{\st{r}{\oc x}} \para
        \store{r}{\oc(\lambda x.x*)}$}
  \label{etrees}
\end{figure}

\subsubsection{Operational Semantics}
The operational semantics of the language is described in
Table~\ref{semantics}.
 \begin{table}[h]
\begin{displaymath}
\begin{array}{cccc}
P\mid P' &\equiv &P'\mid P                                &\mbox{(Commutativity)} \\
(P\mid P')\mid P'' &\equiv &P \mid (P' \mid P'')          &\mbox{(Associativity)} \\
\end{array}
\end{displaymath}
\begin{displaymath}
\begin{array}{lclclcl}
E[(\lambda x.M)V] && &\arrow &E[\sub{M}{V/x}] \\
E[\letm{x}{\bang V}{M}] && &\arrow &E[\sub{M}{V/x}] \\
E[\st{r}{V}]  &&     &\arrow &E[*] &\mid& \store{r}{V} \\
E[\get{r}] &\mid& \store{r}{V} &\arrow &E[V] \\
E[\letm{x}{\get{r}}{M}] &\mid& \store{r}{\bang V} &\arrow &E[\sub{M}{V/x}] &\mid& \store{r}{\bang V}
\end{array}
\end{displaymath}
\caption{Semantics of $\lambdar$ programs}
\label{semantics}
\end{table}
Programs are considered up to a
structural equivalence $\equiv$ which is the least equivalence
relation preserved by static contexts, and which contains the
equations for $\alpha$-renaming and for the commutativity and
associativity of parallel composition. The reduction rules apply
modulo structural equivalence and in a static context $C$.

When writing to a region, values are accumulated rather than
overwritten (remember that $\lambdar$ is an abstract language that can
simulate more concrete ones where values relating to the same region
are associated with distinct addresses).  On the other hand, reading a
region amounts to select non-deterministically one of the values
associated with the region.  We distinguish two rules to read a
region. The first \emph{consumes} the value from the store, like when
reading a communication channel. The second \emph{copies} the value
from the store, like when reading a reference. Note that in this case
the value read must be duplicable (of the shape $\oc V$).

\begin{example}
  \label{program-thread}
  Program $P$ of Figure~\ref{etrees} reduces as follows:
  $$
  \begin{array}{rcl}
    &&\letb{x}{\get{r}}{\st{r}{\oc x}} \para
        \store{r}{\oc(\lambda x.x*)}\\
     &\reduc& \st{r}{\oc(\lambda x.x*)} \para  \store{r}{\oc(\lambda x.x*)}\\
    &\reduc& * \para  \store{r}{\oc(\lambda x.x*)} \para \store{r}{\oc(\lambda x.x*)}
  \end{array}
  $$
\end{example}


\subsection{Extended Depth System}\label{depth-sys-sec}
We start by analysing the interaction between the depth 
of the occurrences and side effects.
We observe that side effects may increase the depth or generate
occurrences at lower depth than the current redex, which violates
Property~\eqref{fact1} and~\eqref{fact2} (see
Section~\ref{subsec-depth-sys}) respectively. Then to find a suitable
notion of depth, it is instructive to consider the following program
examples where $M_r = \letm{z}{\get{r}}{\bang(z*)}$.
$$
  \begin{array}{l}
    (A)\quad E[\st{r}{\bang V}]\\
    (B)\quad\lambda x.\st{r}{x}; \bang \get{r}\\
    (C)\quad \bang(M_r) \para \store{r}{\bang (\lambda y.M_{r'})} \para \store{r'}{\bang (\lambda y.*)}\\
    (D)\quad \bang(M_r) \para \store{r}{\bang (\lambda y.M_r)}
  \end{array}
$$

\begin{description}
\item[$(A)$] Suppose the occurrence
  $\st{r}{\oc V}$ is at depth $\delta > 0$ in $E$. Then
  when evaluating such a term we always end up in a program of the
  shape $E[*] \mid \store{r}{\bang V}$ where the occurrence $\oc V$,
  previously at depth $\delta$, now appears at depth $0$. This
  contradicts Property~\eqref{fact2}.
\item[$(B)$] If we apply
  this program to $\oc V$ we obtain $\bang\bang V$, hence
  Property~\eqref{fact1} is violated because from a program of depth
  $1$, we reduce to a program of depth $2$. We remark that this is
  because the read and write operations do not execute at the same
  depth.
\item[$(C)$] According to our definition, this program has depth $2$,
  however when we reduce it we obtain a term $\bang^3 *$ which has
  depth $3$, hence Property~\eqref{fact1} is violated. This is because
  the occurrence $\lambda y.M_{r'}$ originally at depth $1$ in the
  store, ends up at depth $2$ in the place of $z$ applied to $*$.
\item[$(D)$] If we accept circular stores, we can even write diverging
  programs whose depth is increased by $1$ every two reduction steps.
\end{description}

Given these remarks, the rest of this section is devoted to a \emph{revised
notion of depth} and to depth system extended with side effects.
First, we introduce the following contexts:
$$
\begin{array}{cc}
  \Gamma = x_1:\delta_1, \ldots, x_n:\delta_n &\quad\qquad
  R = r_1:\delta_1, \ldots, r_n:\delta_n
\end{array}
$$
where $\delta_i$ is a natural number. We write
$dom(R)$ for the set
$\set{r_1,\ldots,r_n}$. We write $R(r_i)$ for the depth
$\delta_i$ associated with $r_i$ in the context $R$.

In the sequel, we shall call the notion of depth introduced in
Definition~\ref{depth-functional} \emph{naive depth}.
We revisit the notion of naive depth as follows. 
\begin{definition}[revised depth]
\label{depth-side effect}
Let $P$ be a program, $R$ a region context where $dom(R)$ contains all
the regions of $P$ and $d_n(w)$ the naive depth of an occurrence $w$
of $P$.  If $w$ does not appear under an occurrence $r \leftarrow$ (a
store), then the revised depth $d_r(w)$ of $w$ is $d_n(w)$.
Otherwise, $d_r(w)$ is $R(r) + d_n(w)$.  The revised depth $d_r(P)$ of
the program is the maximum revised depth of its occurrences.
\end{definition}
Note that the revised depth is relative to a fixed region context.  In
the sequel we write $d(\_)$ for $d_r(\_)$.  On functional terms, this
notion of depth is equivalent to the one given in
Definition~\ref{depth-functional}. However, if we consider the program
of Figure~\ref{etrees}, we now have $d(10) = R(r)$ and $d(100) =
d(1000) = d(10000) = d(10001) = R(r) + 1$.

A judgement in the depth system has the shape
$$R;\Gamma \vdash^\delta P$$
and it should be interpreted as follows:
\begin{quotation}
  The free variables of $\oc^{\delta} P$ may only occur at the depth
  specified by the context $\Gamma$, where depths are computed
  according to $R$.
\end{quotation}
The inference rules of the extended depth system are presented in
Table~\ref{depth-system-full}.
\begin{table}[h]
\[
\begin{array}{c}

\inference
{}
{R;\Gamma,x:\delta \vdash^{\delta} x}

\qquad
\inference
{}
{R;\Gamma \vdash^\delta r}

\qquad 

\inference
{}
{R;\Gamma \vdash^\delta *}

\\ \\

\inference
{\fo{x}{M}\leq 1&
 R;\Gamma,x:\delta \vdash^\delta M}
{R;\Gamma \vdash^\delta \lambda x.M}

\qquad 

\inference
{R;\Gamma\vdash^\delta M_i & i=1,2}
{R;\Gamma \vdash^\delta M_1M_2} 

\\ \\

\inference
{R;\Gamma \vdash^{\delta+1} M}
{R;\Gamma \vdash^\delta \oc M}

\qquad 

\inference
{R;\Gamma \vdash^\delta M_1 & R;\Gamma,x:(\delta+1) \vdash^\delta M_2}
{R;\Gamma \vdash^\delta \letb{x}{M_1}{M_2}}

\\\\

\inference
{}
{R,r:\delta;\Gamma \vdash^\delta \get{r}}



\qquad

\inference
{R,r:\delta;\Gamma \vdash^{\delta} V}
{R,r:\delta;\Gamma \vdash^{\delta} \st{r}{V}}

\\\\

\inference
{R,r:\delta;\Gamma \vdash^{\delta} V}
{R,r:\delta;\Gamma \vdash^0 \store{r}{V}}

\qquad

\inference
{R;\Gamma \vdash^\delta P_i & i=1,2}
{R;\Gamma \vdash^\delta (P_1\mid P_2)}

\end{array}
\]
\caption{Depth system for programs: $\lambdard$}
\label{depth-system-full}
\end{table}
We comment on the new rules. A region and the constant $*$ may appear
at any depth. The key cases are those of read and write: the depth of
these two operations is specified by the region context. The current
depth of a store is always $0$, however, the depth of the value in the
store is specified by $R$ (note that it corresponds to the revised
definition of depth). We remark that $R$ is constant in a judgement
derivation.

\begin{definition}[well-formedness]
\label{def-well-formed}
A program $P$ is \emph{well-formed} if for some $R$, $\Gamma$, $\delta$ 
a judgement $R;\Gamma \vdash^\delta P$ can be derived.
\end{definition}

\begin{example}
The program of Figure~\ref{etrees} is well-formed with the following
derivation where  $R(r) = 0$:
\begin{prooftree}
  \AxiomC{$R;\Gamma \vdash^0 \get{r}$}
  \AxiomC{$R;\Gamma, x:1 \vdash^1 x$}
  \UnaryInfC{$R;\Gamma, x:1 \vdash^0 \oc x$}
  \UnaryInfC{$R;\Gamma, x:1 \vdash^0 \st{r}{\oc x}$}
  \BinaryInfC{$R;\Gamma \vdash^0 \letb{x}{\get{r}}{\st{r}{\oc x}}$}
  \AxiomC{$\vdots$}
  \UnaryInfC{$R;\Gamma \vdash^0 \store{r}{\oc (\lambda x.x*)}$}
  \BinaryInfC{$R;\Gamma \vdash^0 \letb{x}{\get{r}}{\st{r}{\oc
        x}} \para \store{r}{\oc (\lambda x.x*)}$}
\end{prooftree}
\end{example}

We reconsider the troublesome programs with side effects.
Program~$(A)$ is well-formed with judgement~$(i)$:
$$
\begin{array}{l l @{\quad} r}
  \label{judgement1}
  R;\Gamma \vdash^0 E[\st{r}{\oc V}] &\textrm{with } R=r:\delta & (i)\\
  \label{judgement2}
  R;\Gamma \vdash^0 \oc M_r \para  \store{r}{\oc(\lambda
    y.M_{r'})} \para \store{r'}{\oc(\lambda y.*)} &\textrm{with }
  R=r:1, r':2 & (ii)
\end{array}
$$
Indeed, the occurrence $\oc V$ is now preserved at depth $\delta$ in
the store. Program~$(B)$ is not well-formed since the read
operation requires $R(r) = 1$ and the write operations require $R(r) =
0$.  Program~$(C)$ is well-formed with judgement~$(ii)$; indeed its
depth does not increase anymore because $\oc M_r$ has depth $2$ but
since $R(r) = 1$ and $R(r') = 2$, $\store{r}{\oc(\lambda y.M_{r'})}$
has depth $3$ and $\store{r'}{\oc(\lambda y.*)}$ has depth $2$. Hence
program~$(C)$ has already depth $3$. Finally, it is worth noticing
that the diverging program~$(D)$ is not well-formed
 since $\get{r}$ appears at depth $1$ in $\oc M_r$ and at
depth $2$ in the store.

\begin{theorem}[properties on the extended depth system]\label{prop-depth-effect}
The following properties hold:
\begin{enumerate}

\item If $R;\Gamma \vdash^\delta M$ and $x$ occurs free in $M$ then 
$x:\delta'$ belongs to $\Gamma$ and all occurrences of $x$ in 
$\bang^\delta M$ are at depth $\delta'$.

\item If $R;\Gamma \vdash^\delta P$ then $R;\Gamma,\Gamma'\vdash^\delta P$.

\item If $R;\Gamma,x:\delta' \vdash^\delta M$ and $R;\Gamma \vdash^{\delta'} V$ 
then $R;\Gamma \vdash^\delta \sub{M}{V/x}$ and \\
$d(\bang^\delta \sub{M}{V/x}) \leq \w{max}(d(\bang^\delta M),d(\bang^{\delta'} V))$.

\item If $R;\Gamma \vdash^0 P$ and $P\arrow P'$ then $R;\Gamma \vdash^0 P'$ and 
$d(P) \geq d(P')$.

\end{enumerate}
\end{theorem}

\subsection{Elementary Bound}\label{elem-bound-sec}
In this section, we prove that well-formed programs terminate in
elementary time.  
The measure of Definition~\ref{def-measure} extends trivially to programs
except that to simplify the proofs of the following properties, we assume the
occurrences labelled with $\para$ and $r \leftarrow$ do not count in the
measure and that $\mathsf{set}(r)$ counts for two occurrences such
that the measure strictly decreases on the rule $E[\st{r}{V}] \reduc
E[*] \para \store{r}{V}$.

We derive a similar termination property:
\begin{proposition}[termination]
\label{prop-termination}
If $P$ is well-formed, $P \arrow P'$ and $n\geq d(P)$ then $\mu_n(P)>\mu_n(P')$.
\end{proposition}
\begin{proof} 
  By a case analysis on the new reduction rules.

\begin{itemize}
\item $P\equiv E[\st{r}{V}]  \arrow P' \equiv E[*] \mid
  \store{r}{V}$\\
If $R;\Gamma \vdash^{\delta}  \st{r}{V}$ then by~\ref{prop-depth-effect}(4) we have $R;\Gamma
\vdash^0 \store{r}{V}$ with $R(r) =
\delta$. Hence, by definition of the depth, the occurrences
in $V$ stay at depth $\delta$ in $\store{r}{V}$. However, the node
$\st{r}{V}$ disappears, and both $*$ and $\store{r}{V}$ are null
occurrences, thus $\occ_{\delta}(P') = \occ_{\delta}(P) -
1$. The number of occurrences at other depths stay unchanged, hence
$\mu_n(P) > \mu_n(P')$.

\item $P \equiv E[\get{r}] \mid \store{r}{V} \arrow P'
  \equiv E[V]$ \\
  If $R;\Gamma \vdash^0 \store{r}{V}$ with $R(r) = \delta$, then
  $\get{r}$ must be at depth $\delta$ in $E[~]$. Hence, by definition
  of the depth, the occurrences in $V$ stay at depth $\delta$, while
  the node $\get{r}$ and $\para$
  disappear. Thus $\occ_{\delta}(P') = \occ_{\delta}(P) - 1$ and the
  number of occurrences at other depths stay unchanged, hence
  $\mu_n(P) > \mu_n(P')$.

\item $P \equiv E[\letm{x}{\get{r}}{M}] \mid \store{r}{\bang V} \arrow
P' \equiv E[\sub{M}{V/x}] \mid \store{r}{\bang V}$\\
This case is the only source of duplication with the reduction rule
on $\letbang$.
Suppose $R;\Gamma \vdash^{\delta} \letm{x}{\get{r}}{M}$. Then
we must have $R;\Gamma \vdash^{\delta+1} V$.
The restrictions on the formation of terms require that
$x$ may only occur in $M$ at depth $1$ and hence in $P$
at depth $\delta+1$. Hence the occurrences in $V$ stay at the same
depth in $\sub{M}{V/x}$, while
the $\mathsf{let}$, $\get{r}$ and some $x$ nodes disappear, hence
$\occ_{\delta}(P) \leq \occ_{\delta}(P') -2$. The number of occurrences of $x$ in $M$ is bound by 
$k=\occ_{\delta+1}(P) \geq 2$. 
Thus if  $j>\delta$  then $\occ_{j}(P') \leq k \cdot \occ_j(P)$.
Clearly,  $\occ_j(M) = \occ_j(M')$ if $j<i$. Hence, we have
\begin{equation}
  \label{na-letget-decrease}
  \mu_n(P') \leq (\mu_n^{i+1}(P) \cdot k, \occ_i(P)-2, \mu_{i-1}(P))
\end{equation}
and $\mu_n(P) > \mu_n(P')$. 
\end{itemize}
\end{proof}

Then we have the following theorem.
\begin{theorem}[elementary bound]
\label{elementary-bound-effect}
Let $P$ be a well-formed program with $\alpha =d(P)$ and let
$t_\alpha$ denote the tower function with $\alpha+1$ arguments.  Then
if $P\arrow P'$ then $t_\alpha(\mu_\alpha(P)) >
t_\alpha(\mu_\alpha(P'))$.
\end{theorem}
\begin{proof}
  From the proof of termination, we remark that the only new rule that
  duplicates occurrences is the one that copies from the
  store. Moreover, the derived inequality~\eqref{na-letget-decrease} is exactly
  the same as the inequality~\eqref{na-letbang-decrease}. Hence the
  arithmetic of the proof is exactly the same as in the proof of elementary bound
  for the functional case.
\end{proof}

\begin{corollary}
The normalisation of programs of bounded depth can be performed
in time elementary in the size of the terms.
\end{corollary}

\section{An Elementary Affine Type System}\label{type-system-sec}
The depth system  entails termination in elementary
time but does {\em not} guarantee that programs `do not go wrong'. 
In particular, the introduction and elimination of bangs during evaluation may
generate programs that deadlock, \emph{e.g.},
\begin{equation}
\label{deadlock}
\letb{y}{(\lambda x.x)}{\oc(yy)}
\end{equation}
is well-formed but the evaluation is stuck.  
In this section we introduce an elementary affine type system ($\lambdart$) that guarantees
that programs cannot deadlock (except when trying to read an empty store).

The upper part of Table~\ref{types-contexts} introduces the syntax
of types and contexts.
\begin{table}[h]
  \centering
\[
\begin{array}{rcll}
\multicolumn{3}{l}{t,t',\ldots} & \mbox{(Type variables)}\\
\alpha&::=& \behtype \mid A                             &\mbox{(Types)} \\
A&::=& t \mid \tertype \mid A\multimap \alpha \mid \oc A \mid \forall t.A
\mid \rgtype{r}{A}    &\mbox{(Value-types)} \\
\Gamma &::=&
\hyp{x_{1}}{\delta_{1}}{A_{1}},\ldots,\hyp{x_{n}}{\delta_{n}}{A_{n}}
&\mbox{(Variable contexts)} \\
R&::=& r_1:(\delta_1,A_1),\ldots,r_n:(\delta_n,A_n)  &\mbox{(Region contexts)} \\
\end{array}
\]
\\
\[
\begin{array}{c}
\inference
{}
{R \downarrow t}

\qquad

\inference
{}
{R \downarrow \tertype}
\qquad

\inference
{}
{R \downarrow \behtype}
\qquad

\inference
{R\downarrow A & R\downarrow \alpha}
{R\downarrow (A\limp \alpha)}

\\\\
\inference
{R\downarrow A}
{R\downarrow \oc A}

\qquad

\inference
{r:(\delta,A) \in R}
{R\downarrow \rgtype{r}{A}}

\qquad

\inference
{R \downarrow A & t \notin R}
{R \downarrow \forall t.A}

\\\\
\inference
{\forall r:(\delta,A) \in R & R\downarrow A}
{R\vdash}

\qquad

\inference
{R\vdash & R\downarrow \alpha}
{R\vdash \alpha} 

\\\\
\inference
{\forall \hyp{x}{\delta}{A}\in \Gamma & R\vdash A }
{R\vdash \Gamma}

\end{array}
\]
\caption{Types and contexts}\label{types-contexts}
\end{table}
Types are denoted with $\alpha,\alpha',\ldots$. Note that we
distinguish a special behaviour type $\behtype$ which is given to the
entities of the language which are not supposed to return a value
(such as a store or several terms in parallel) while types of entities
that may return a value are denoted with $A$.  Among the types $A$, we
distinguish type variables $t, t', \ldots$, a terminal type
$\tertype$, an affine functional type $A \limp \alpha$, the type
$\bang A$ of terms of type $A$ that can be duplicated, the type
$\forall t.A$ of polymorphic terms and the type $\rgtype{r}{A}$ of the
region $r$ containing values of type $A$. Hereby types may depend on
regions.

In contexts, natural numbers $\delta_i$ play the same role as in the
depth system. Writing $x : (\delta, A)$ means that the variable $x$
ranges on values of type $A$ and may occur at depth $\delta$. Writing
$r:(\delta,A)$ means that addresses related to region $r$ contain
values of type $A$ and that read and writes on $r$ may only happen at
depth $\delta$. The typing system will additionally guarantee that
whenever we use a type $\rgtype{r}{A}$ the region context contains an
hypothesis $r:(\delta,A)$.

Because types depend on regions, we have to be careful in stating in
Table~\ref{types-contexts} when a region-context and a type are
compatible ($R\downarrow \alpha$), when a region context is
well-formed ($R\vdash$), when a type is well-formed in a region
context (\mbox{$R\vdash \alpha$}) and when a context is well-formed in
a region context (\mbox{$R\vdash \Gamma$}).  A more informal way to
express the condition is to say that a judgement
$r_{1}:(\delta_1,A_{1}),\ldots,r_{n}:(\delta_n,A_{n}) \vdash \alpha$ is well formed provided
that:
$(1)$ all the region names
occurring in the types $A_1,\ldots,A_n,\alpha$ belong to the set
$\set{r_1,\ldots,r_n}$, 
$(2)$ all types of the shape
$\rgtype{r_{i}}{B}$ with $i\in \set{1,\ldots,n}$ and occurring in the
types $A_1,\ldots,A_n,\alpha$ are such that \mbox{$B=A_i$}.
We notice the following substitution property on types.

\begin{proposition}
\label{type-substitution}
If $R\vdash \forall t.A$ and $R\vdash B$ then 
$R\vdash \sub{A}{B/t}$.
\end{proposition}

\begin{example}
One may verify that $r:(\delta,\tertype \limp \tertype)
\vdash \rgtype{r}{(\tertype \limp \tertype)}$ can be derived while
the following judgements cannot:
$r:(\delta,\tertype) \vdash \rgtype{r}{(\tertype \limp \tertype)}$, 
$r:(\delta,\rgtype{r}{\tertype}) \vdash \tertype$.
\end{example}

A typing judgement takes the form:
$$
R;\Gamma \vdash^{\delta} P:\alpha
$$
It attributes a type $\alpha$ to the program $P$ at depth $\delta$, in
the region context $R$ and the context $\Gamma$.
Table~\ref{eal-type-system} introduces an elementary affine type system 
{\em with regions}.
\begin{table}[h]
$$
\begin{array}{c}
  
\inference
{R\vdash \Gamma \\  \hyp{x}{\delta}{A}\in \Gamma}
{R;\Gamma \vdash^\delta x:A}

\qquad

\inference
{R\vdash \Gamma}
{R;\Gamma \vdash^\delta *:\tertype} 

\qquad

\inference
{R\vdash \Gamma \\  r:(\delta',A) \in R}
{R;\Gamma \vdash^\delta r:\rgtype{r}{A}}\\ \\

\inference
{\fo{x}{M} \leq 1 \\ R;\Gamma,\hyp{x}{\delta}{A} \vdash^\delta M:\alpha}
{R;\Gamma \vdash^\delta \lambda x.M:A \multimap \alpha}

\qquad

\inference{
R;\Gamma\vdash^\delta M:A \multimap \alpha &
R;\Gamma\vdash^\delta N:A}
{R;\Gamma \vdash^\delta MN:\alpha} \\ \\ 

\inference
{R;\Gamma \vdash^{\delta+1} M:A}
{R;\Gamma \vdash^\delta \oc M:\oc A}

\quad

\inference{
R;\Gamma\vdash^\delta M:\oc A& 
R;\Gamma,\hyp{x}{\delta+1}{A} \vdash^\delta N:B}
{R;\Gamma \vdash^\delta \letb{x}{M}{N}:B}

\\ \\

\inference
{R;\Gamma \vdash^\delta M : A & t \notin (R;\Gamma)}
{R;\Gamma \vdash^\delta M: \forall t.A}

\qquad

\inference
{R;\Gamma \vdash^\delta M : \forall t.A & R \vdash B}
{R;\Gamma \vdash^\delta M : \sub{A}{B/t}}

\\\\

\inference
{r:(\delta,A) \in R \\ R\vdash\Gamma}
{R;\Gamma \vdash^\delta \get{r}:A}
\quad

\inference
{r:(\delta,A) \in R \\ R;\Gamma \vdash^\delta
V : A}
{R;\Gamma \vdash^\delta \st{r}{V}:\tertype} 

\quad

\inference
{r:(\delta,A) \in R \\ R;\Gamma \vdash^\delta V:A}
{R;\Gamma \vdash^0  \store{r}{V} : \behtype}

\\\\

\inference
{R;\Gamma \vdash^\delta P:\alpha &
R;\Gamma \vdash^\delta S:\behtype}
{R;\Gamma \vdash^\delta (P\mid S):\alpha}  

\qquad

\inference
{ P_i \mbox{ not a store} \; i=1,2\\
  R;\Gamma \vdash^\delta P_i:\alpha_i}
{R;\Gamma \vdash^\delta (P_1\mid P_2):\behtype}  
\end{array}
$$
\caption{An elementary affine type system: $\lambdart$}
\label{eal-type-system}
\end{table}
One can see that the $\delta$'s are treated as in the depth
system.  Note that a region $r$ may occur at any depth. In the
$\letbang$ rule, $M$ should be of type $\oc A$ since $x$ of type $A$
appears one level deeper. A program in parallel with a store should
have the type of the program since we might be interested in the value
the program reduces to; however, two programs in parallel cannot
reduce to a single value, hence we give them a behaviour type. The
polymorphic rules are straightforward where $t \notin (R;\Gamma)$
means $t$ does not occur free in a type of $R$ or $\Gamma$.

\begin{example}
  The well-formed program~$(C)$ can be given the following typing
  judgement:
$
R; \_ \vdash^0 \bang(M_r) \mid \store{r}{\bang (\lambda y.M_{r'})} \mid
\store{r'}{\bang (\lambda y.*)} : \oc\oc\tertype
$
where: $R = r:(1,\oc(\tertype \multimap \tertype)), r':(2,\oc(\tertype \multimap
\tertype))$.
Also, we remark that the deadlocking program~\eqref{deadlock}
admits no typing derivation.
\end{example}

\begin{theorem}[subject reduction and progress] \label{progress-thm}
The following properties hold.
\begin{enumerate}
\item\emph{(Well-formedness)} Well-typed
  programs are well-formed. 
\item\emph{(Weakening)} \label{weakening} If $R;\Gamma \vdash^\delta
  P:\alpha$ and $R\vdash \Gamma, \Gamma'$ then $R;\Gamma,
  \Gamma'\vdash^\delta P:\alpha$.
\item\emph{(Substitution)} \label{sub-lemma} If
  $R;\Gamma,\hyp{x}{\delta'}{A} \vdash^\delta M:\alpha$ and
  $R;\Gamma'\vdash^{\delta'} V: A$ and $R \vdash \Gamma, \Gamma'$ then
  $R;\Gamma, \Gamma'\vdash^\delta \sub{M}{V/x}:\alpha$.
\item\emph{(Subject Reduction)} \label{sub-red-thm} If $R;\Gamma
  \vdash^\delta P:\alpha$ and $P\arrow P'$ then $R;\Gamma
  \vdash^\delta P':\alpha$.
\item\emph{(Progress)} Suppose $P$ is a closed typable program which cannot reduce. Then
$P$ is structurally equivalent to a program
\[
M_1 \para \cdots \para M_m \para S_1 \mid \cdots
\para S_n\quad m,n\geq 0
\]
where $M_i$ is either a value or can be decomposed as a term
$E[\get{r}]$ such that no value is
associated with the region $r$ in the stores $S_1,\ldots,S_n$.
\end{enumerate}
\end{theorem}

\section{Expressivity} \label{sec-expressivity} In this section, we
consider two results that illustrate the expressivity of the
elementary affine type system.  First we show that all elementary
functions can be represented and second we develop an example of
iterative program with side effects.

\subsection{Completeness}
The representation result just relies on the functional core of the
language $\lambdaft$.  Building on the standard concept of Church
numeral, Table~\ref{church-encodings} provides a representation for
natural numbers and the multiplication function.
\begin{table}[h]
\[
\begin{array}{r c l l}
\nat &=& \forall t.\bang (t\limp t) \limp \bang (t\limp t)  &\mbox{(type of numerals)}\\\\
\churchn{n}&:& \nat         &\mbox{(numerals)}\\
\churchn{n}&=& \lambda f.\letm{f}{f}{\bang(\lambda x.f(\cdots (fx) \cdots))}\\\\


\mult &:& \nat \limp (\nat\limp \nat) &\mbox{(multiplication)}\\ 
\mult&=& \lambda n.\lambda m. \lambda f. \s{let} \ \bang f = f \
\s{in}\ 
    n(m \bang f)



\end{array}
\]
\caption{Representation of natural numbers and the multiplication function}\label{church-encodings}
\end{table}
We denote with $\mathbb{N}$ the set of natural numbers. The
precise notion of representation is spelled out in the following
definitions where by strong $\beta$-reduction we mean that reduction
under $\lambda$'s is allowed.

\begin{definition}[number representation]
  Let $\emptyset \vdash^\delta M : \nat$. We say $M$ \emph{represents} $n \in \mathbb{N}$, written $M
  \Vdash n$, if, by using a strong $\beta$-reduction relation, $M \treduc  \churchn{n}$.
\end{definition}
\begin{definition}[function representation]
  \label{def-function-rep}
  Let $\emptyset \vdash^\delta F: (\nat_1 \multimap \ldots \multimap \nat_k)
  \multimap \oc^p\nat$ where $p
\geq 0$ and $f:\mathbb{N}^k \rightarrow \mathbb{N}$. We say
$F$ \emph{represents} $f$, written $F \Vdash f$,
if for all $M_i$ and $n_i \in \mathbb{N}$  where $1 \leq i \leq
k$ such that $\emptyset \vdash^\delta M_i:N$ and $M_i \Vdash n_i$,
$FM_1\ldots M_k \Vdash f(n_1,\ldots,n_k)$. 
\end{definition}

Elementary functions are also characterized as the smallest class of
functions containing zero, successor, projection, subtraction and
which is closed by composition and bounded summation/product.  These
functions can be represented in the sense of
Definition~\ref{def-function-rep} by adapting the proofs from Danos
and Joinet~\cite{Danos}.
\begin{theorem}[completeness]\label{complete-thm}
  Every elementary function is representable in $\lambdaft$.
\end{theorem}

\subsection{Iteration with Side Effects}
We rely on a slightly modified language where 
reads, writes and stores
relate to concrete addresses rather than to abstract regions.
In particular, we introduce terms of the form $\nu x \ M$ to generate a fresh
address name $x$ whose scope is $M$. One can then write the following
program:
$$
\nu x \ ((\lambda y.\st{y}{V})x) \treduc \nu x \ * \mid \store{x}{V}
$$
where $x$ and $y$ relate to a region $r$, \ie they
are of type $\regtype{r}{A}$. Our type system can be easily 
adapted by associating region types with the address names.
%
Next we show that it is possible to
program the iteration of operations producing a side effect on an
inductive data structure. Specifically, in the following we show how
to iterate, possibly in parallel, an update operation on a list of
addresses of the store. The examples have been tested on a 
running implementation of the language.

Following Church encodings, we define the representation of lists and
the associated iterator in Table~\ref{list-representation}.
\begin{table}[h]
    $$
    \begin{array}{r c l l}
    \listtype{A} &=& \forall t.\oc(A \multimap t \multimap t) \multimap \oc(t
    \multimap t) & \text{(type of lists)}\\\\

        \churchl{u_1,\ldots, u_n} & : & \listtype{A} &
    \text{(list represent.)}\\
    \churchl{u_1,\ldots, u_n} &=& \lambda f.\letb{f}{f}{}
     \oc (\lambda x.fu_1(fu_2\ldots(fu_nx)) \\\\

     \lit &:& \forall
     u. \forall t. \oc(u
     \multimap t \multimap t) \multimap \listtype{u} \multimap \oc t
     \multimap \oc t & \text{(iterator)}\\
    \lit & = & \lambda f.\lambda l.\lambda z.\letb{z}{z}{}
     \letb{y}{lf}{\oc(yz)}


  \end{array}
  $$
  \caption{Representation of lists}
  \label{list-representation}
\end{table}
Here is the function multiplying the numeral pointed by an address
  at region $r$:
$$
\begin{array}{r c l}
  \mathsf{update} &:& \oc\regtype{r}{\nat} \multimap \oc \tertype \multimap \oc \tertype\\
  \mathsf{update} &=& \lambda x.\letb{x}{x}{\lambda
    z. {\oc ((\lambda y.\st{x}{y})(\mult \; \churchn{2}\; \get{x})}}
\end{array}
$$
Consider the following list of addresses and stores:
 $$
  \churchl{\oc x,\oc y,\oc z} \para \store{x}{ \churchn{m}} \para
  \store{y}{ \churchn{n}} \para \store{z}{ \churchn{p}}
  $$
  Note that the bang constructors are needed to match the type
  $\oc\regtype{r}{ \nat}$ of the
  argument of $\mathsf{update}$.
  Then we define the iteration as:
  $$
\begin{array}{c}
  \mathsf{run} : \oc\oc \tertype\qquad
  \mathsf{run} = \lit \; \oc \mathsf{update} \;
  \churchl{\oc x,\oc y,\oc z} \; \oc\oc * 
\end{array}
  $$
  Notice that it is well-typed with $R = r:(2,\nat)$ since both the
  read and the write appear at depth $2$. Finally, the program reduces
  by updating the store as expected:
  $$
  \begin{array}{rcl}
  &&\mathsf{run} \para  \store{x}{ \churchn{m}} \para
  \store{y}{\churchn{n}} \para \store{z}{ \churchn{p}}\\
  &\treduc& \oc\oc \tertype \para \store{x}{ \churchn{2m}} \para
  \store{y}{ \churchn{2n}} \para \store{z}{ \churchn{2p}}
  \end{array}
  $$
Building on this  example, suppose we want to write a program
with three concurrent threads where each thread 
multiplies by $2$ the memory cells pointed by a list.
Here is a function waiting to apply a functional $f$ to a value $x$ in
  three concurrent threads:
  $$
  \begin{array}{r c l}
    \mathsf{gen\_threads} &:& \forall t.\forall t'.\oc(t \multimap t')
    \multimap \oc t \multimap \behtype\\
    \mathsf{gen\_threads} &=& \lambda f.\letb{f}{f}{\lambda
      x.\letb{x}{x}{\oc(fx) \para \oc (fx) \para \oc (fx)}}
  \end{array}
  $$
  We define the functional $F$ as $\mathsf{run}$ but parametric in the
  list:
  $$
  \begin{array}{l}
    F : \listtype{\oc\regtype{r}{\nat}} \multimap \oc\oc \tertype\qquad
    F = \lambda l.\lit \; \oc \mathsf{update} \; l \; \oc\oc *
  \end{array}
  $$
  And the final term is simply:
  $$
  \begin{array}{c}
    \mathsf{run\_threads} : \behtype\qquad
    \mathsf{run\_threads} = \mathsf{gen\_threads} \; \oc F \;
    \oc\churchl{\oc x,\oc y,\oc z}
  \end{array}
  $$
  where $R=r:(3,\oc\nat)$. Our program then reduces as follows:
  $$
  \begin{array}{c l c l c l c l}
  &\mathsf{run\_threads} &\para&  \store{x}{ \churchn{m}} &\para&
  \store{y}{ \churchn{n}} &\para& \store{z}{ \churchn{p}}\\
  \treduc& \oc\oc\oc \tertype \para \oc\oc\oc \tertype \para \oc\oc\oc \tertype &\para& \store{x}{ \churchn{8m}} &\para&
  \store{y}{ \churchn{8n}} &\para& \store{z}{ \churchn{8p}}
  \end{array}
  $$
Note that different thread interleavings are possible but in this
particular case the reduction is confluent.

\section{Conclusion}
We have introduced a type system for a higher-order functional language with
multithreading and side effects that guarantees termination in elementary time
thus providing a significant extension of previous work that had focused 
on purely functional programs.
In the proposed approach, the depth system plays a key role and allows for
a relatively simple presentation. In particular we notice that we can dispense
both with the notion of {\em stratified region} that arises in recent work on the
termination of higher-order programs with side effects \cite{Amadio09,Boudol10} and with
the distinction between affine and intuitionistic hypotheses
\cite{Barber96,air-hal}.

As a future work, we would like to adapt our approach to polynomial
time. In another direction, one could ask if it is possible to program
in a simplified language without bangs and then try to infer types or
depths.



\paragraph{Acknowledgements}
We would like to thank Patrick Baillot for numerous helpful
discussions and a careful reading on a draft version of this
report.

\bibliographystyle{abbrv} 
\bibliography{madet-elementary}

\newpage
\appendix

\section{Proofs}\label{proofs-sec}

\subsection{Proof of theorem \ref{prop-depth-effect}}

  \begin{enumerate}
\item 
We consider the last rule applied in the typing of $M$.

\begin{itemize}

\item $\Gamma,x:\delta\vdash^\delta x$. The only free variable is $x$ and indeed
it is at depth $\delta$ in $\bang^\delta x$.

\item $\Gamma \vdash^\delta \lambda y.M$ is derived from $\Gamma,y:\delta \vdash^\delta M$.
If $x$ is free in $\lambda y.M$ then $x\neq y$ and $x$ is free in $M$. 
By inductive hypothesis, $x:\delta' \in \Gamma,y:\delta$ and all occurrences
of $x$ in $\bang^\delta M$ are at depth $\delta'$. By definition of depth, the same is true
for $\bang^\delta (\lambda y.M)$.

\item $\Gamma \vdash^\delta (M_1M_2)$  is derived from $\Gamma \vdash^\delta M_i$ for $i=1,2$.
By inductive hypothesis, $x:\delta' \in \Gamma$ and all occurrences of $x$ in $\bang^\delta M_i$,
$i=1,2$ are at depth $\delta'$. By definition of depth, the same is true for
$\bang^\delta (M_1M_2)$.

\item $\Gamma \vdash^\delta \bang M$ is derived from $\Gamma \vdash^{\delta+1} M$.
By inductive hypothesis, $x:\delta'\in \Gamma$ and all occurrences of $x$ in 
$\bang^{\delta+1} M$ are at depth $\delta'$ and notice that 
$\bang^{\delta+1}M = \bang^\delta(\bang M)$.

\item $\Gamma \vdash^\delta \letm{y}{M_{1}}{M_{2}}$ is derived from
$\Gamma \vdash^\delta M_1$ and $\Gamma,y:(\delta+1) \vdash^\delta M_2$.
Without loss of generality, assume $x\neq y$. 
By inductive hypothesis, $x:\delta' \in \Gamma$ and all occurrences of $x$ in 
$\bang^\delta M_i$, $i=1,2$ are at depth $\delta'$. By definition of depth,
the same is true for $\bang^\delta (\letm{y}{M_{1}}{M_{2}})$.

\item $M \equiv *$ or $M \equiv r$ or $M \equiv \get{r}$.
There is no free variable in these terms.

\item $M \equiv \letb{y}{\get{r}}{N}$.
  We have 
  $$
  \inference
  {R,r:\delta;\Gamma,y:(\delta + 1) \vdash^\delta N}
  {R,r:\delta;\Gamma \vdash^\delta \letb{y}{\get{r}}{N}}
  $$
  If $x$ occurs free in $M$ then $x$ occurs free in $N$. By
  induction hypothesis, $x:\delta' \in \Gamma$ and all occurrences of
  $x$ in $\oc^{\delta}N$ are at depth  $\delta'$. By
  definition of the  depth, this is also true for $\oc^{\delta}( \letb{y}{\get{r}}{N})$.

\item $M \equiv \st{r}{V}$.
  We have 
  $$
  \inference
  {R,r:\delta;\Gamma \vdash^{\delta} V}
  {R,r:\delta;\Gamma \vdash^{\delta} \st{r}{V}}
  $$
  If $x$ occurs free in $\st{r}{V}$ then $x$ occurs free in $V$. By
  induction hypothesis, $x:\delta' \in \Gamma$ and all occurrences of
  $x$ in $\oc^{\delta}V$ are at depth $\delta'$. By definition of the
  depth, this is also true for $\oc^{\delta}(\st{r}{V})$.

\item $M \equiv (M_1 \mid M_2)$.
  We have
  $$
  \inference
  {R;\Gamma \vdash^\delta M_i & i=1,2}
  {R;\Gamma \vdash^\delta (M_1\mid M_2)}
  $$
  If $x$ occurs free in $M$ then $x$ occurs free in $M_i$, $i=1,2$.
  By induction hypothesis, $x:\delta' \in \Gamma$ and all occurrences
  of $x$ in $\oc^\delta M_i$, $i=1,2$, are at depth $\delta'$.
  By definition of depth, the same is true of $\oc^\delta(M_1\mid M_2)$.

\end{itemize}

\item All the rules can be weakened by adding a context $\Gamma'$.

\item If $x$ is not free in $M$, we just have to check that
any proof of $\Gamma,x:\delta' \vdash^\delta M$ can be transformed 
into a proof of $\Gamma \vdash^\delta M$. 

So let us assume $x$ is free in $M$. 

We consider first the bound on the depth. By (1), we know that all
occurrences of $x$ in $\bang^\delta M$ are at depth $\delta'$.  By
definition of depth, it follows that $\delta'\geq \delta$ and the
occurrences of $x$ in $M$ are at depth $(\delta'-\delta)$.  An
occurrence in $\bang^{\delta'} V$ at depth $\delta'+\delta''$ will
generate an occurrence in $\bang^\delta \sub{M}{V/x}$ at the same depth
$\delta+ (\delta'-\delta)+\delta''$.

Next, we proceed by induction on the derivation of $\Gamma,x:\delta'\vdash^\delta M$.

\begin{itemize}

\item $\Gamma,x:\delta \vdash^\delta x$. Then $\delta=\delta'$, $\sub{x}{V/x}=V$,
and by hypothesis $\Gamma \vdash^{\delta'} V$.

\item $\Gamma,x:\delta' \vdash^{\delta} \lambda y.M$ 
is derived from $\Gamma,x:\delta',y:\delta \vdash^\delta M$,  with $x\neq y$ 
and $y$ not occurring in $N$.
By (2), $\Gamma,y:\delta\vdash^{\delta'} V$. 
By inductive hypothesis, $\Gamma,y:\delta \vdash^\delta \sub{M}{V/x}$, 
and then we conclude $\Gamma \vdash^\delta \sub{(\lambda y.M)}{V/x}$.

\item $\Gamma,x:\delta'\vdash^\delta (M_1M_2)$ is derived from
$\Gamma,x:\delta'\vdash^\delta M_i$, for $i=1,2$.
By inductive hypothesis, $\Gamma \vdash^\delta \sub{M_i}{V/x}$, for $i=1,2$ 
and then we conclude $\Gamma \vdash^\delta \sub{(M_1M_2)}{V/x}$.

\item $\Gamma,x:\delta' \vdash^\delta \bang M$ is derived from
$\Gamma,x:\delta' \vdash^{\delta+1} M$. 
By inductive hypothesis, $\Gamma \vdash^{\delta+1} \sub{M}{V/x}$, and then we conclude
$\Gamma \vdash^{\delta} \sub{\bang M}{V/x}$.

\item $\Gamma,x:\delta' \vdash^\delta \letm{y}{M_1}{M_2}$, with 
$x\neq y$ and $y$ not free in $V$ is derived from
$\Gamma,x:\delta' \vdash^\delta M_1$ and $\Gamma,x:\delta',y:(\delta+1) \vdash^\delta M_2$.
By inductive hypothesis, $\Gamma \vdash^\delta \sub{M_1}{V/x}$ 
$\Gamma,y:(\delta+1) \vdash^\delta \sub{M_2}{V/x}$, and then we conclude
$\Gamma \vdash^\delta \sub{(\letm{y}{M_1}{M_2})}{V/x}$.

\item $M \equiv  \letb{y}{\get{r}}{M_1}$.
  We have 
  $$
  \inference
  {R,r:\delta;\Gamma,x:\delta',y:(\delta + 1) \vdash^\delta M_1}
  {R,r:\delta;\Gamma,x:\delta' \vdash^\delta \letb{y}{\get{r}}{M_1}}
  $$
  By induction hypothesis we get 
  $$
  R,r:\delta;\Gamma,y:(\delta + 1) \vdash^\delta \sub{M_1}{V/x}
  $$
  and hence we derive 
  $$
  R,r:\delta;\Gamma \vdash^\delta \sub{(\letb{y}{\get{r}}{M_1})}{V/x}
  $$

\item $M \equiv \st{r}{V'}$.
  We have
  $$
  \inference
  {R,r:\delta;\Gamma,x:\delta' \vdash^{\delta} V'}
  {R,r:\delta;\Gamma,x:\delta'  \vdash^{\delta} \st{r}{V'}}
  $$
  By induction hypothesis we get
  $$
  R,r:\delta;\Gamma \vdash^{\delta} \sub{V'}{V/x}
  $$
  and hence we derive
  $$
  R,r:\delta;\Gamma  \vdash^{\delta} \sub{(\st{r}{V'})}{V/x}
  $$

\item $M \equiv (M_1 \mid M_2)$.
  We have 
  $$
  \inference
  {R;\Gamma,x:\delta' \vdash^\delta M_i & i=1,2}
  {R;\Gamma,x:\delta' \vdash^\delta (M_1\mid M_2)}
  $$
  
  By induction hypothesis we derive 
  $$
  R;\Gamma \vdash^\delta \sub{M_i}{V/x}
  $$
  and hence we derive 
  $$
  R;\Gamma \vdash^\delta \sub{(M_1\mid M_2)}{V/x}
  $$
 
\end{itemize}

\item
We proceed by case analysis on the reduction rules.
\begin{itemize}

\item Suppose $\Gamma \vdash^0 E[(\lambda x.M)V]$. 
Then for some $\Gamma'$ extending $\Gamma$ and $\delta\geq 0$ we must have
$\Gamma' \vdash^\delta (\lambda x.M)V$.
This must be derived from
$\Gamma',x:\delta \vdash^\delta M$ and 
$\Gamma' \vdash^\delta V$.
By (3), with $\delta=\delta'$, it follows that 
$\Gamma'\vdash^\delta \sub{M}{V/x}$ and that 
the depth of an occurrence in $E[\sub{M}{V/x}]$ 
is bounded by the depth of an occurrence which is already
in $E[(\lambda x.M)V]$.
Moreover, we  can derive $\Gamma \vdash^0 E[\sub{M}{V/x}]$.

\item  Suppose $\Gamma \vdash^0 E[\letm{x}{\bang V}{M}]$. 
Then for some $\Gamma'$ extending $\Gamma$ and $\delta\geq 0$ we must have
$\Gamma' \vdash^\delta \letm{x}{\bang V}{M}$.
This must be derived from
$\Gamma',x:(\delta+1) \vdash^\delta M$ and 
$\Gamma' \vdash^{(\delta+1)} V$.
By (3), with $(\delta+1)=\delta'$, it follows that 
$\Gamma'\vdash^\delta \sub{M}{V/x}$ and that 
the depth of an occurrence in $E[\sub{M}{V/x}]$ 
is bounded by the depth of an occurrence which is already
in $E[\letm{x}{\bang V}{M}]$.
Moreover, we  can derive $\Gamma \vdash^0 E[\sub{M}{V/x}]$.

\item $E[\st{r}{V}]  \arrow  E[*] \mid \store{r}{V}$\\
We have 
$R;\Gamma \vdash^0 E[\st{r}{V}]$
from which we derive 
$$
\inference
{R;\Gamma \vdash^\delta V}
{R;\Gamma \vdash^\delta \st{r}{V}}
$$
for some $\delta \geq 0$, with $r:\delta \in R$. Hence we can derive
$$
\inference
{R;\Gamma \vdash^{\delta} V}
{R;\Gamma \vdash^0 \store{r}{V}}
$$
Moreover, we have 
as an axiom $R;\Gamma \vdash^{\delta} *$ 
thus we can derive
$R;\Gamma \vdash^0 E[*]$.
Applying the parallel rule we finally get
$$R;\Gamma \vdash^0 E[*] \mid \store{r}{V}$$

Concerning the depth bound, clearly we have $d(E[*] \mid \store{r}{V})
= d(E[\st{r}{V}])$.

\item $E[\get{r}] \mid \store{r}{V} \rightarrow E[\sub{M}{V/x}]$\\
We have $R;\Gamma \vdash^0 E[\get{r}] \mid \store{r}{V}$ from which we
derive 
$$
\inference
{}
{R;\Gamma \vdash^{\delta} \get{r}}
$$
and 
$$
\inference
{}
{R;\Gamma,x:\delta \vdash^{\delta} M}
$$
for some $\delta \geq 0$, with $r:\delta \in R$, and 
$$
\inference
{R;\Gamma \vdash^{\delta} V}
{R;\Gamma \vdash^0 \store{r}{V}}
$$
Hence we can derive
$$
R;\Gamma \vdash^0 E[V]
$$

Concerning the depth bound, clearly we have $d(E[V])
= d(E[\get{r}] \mid \store{r}{V})$.

\item $E[\letb{x}{\get{r}}{M}] \mid \store{r}{\oc V} \rightarrow
  E[\sub{M}{V/x}] \mid \store{r}{\oc V}$\\
  We have  $R;\Gamma \vdash^0 E[\letb{x}{\get{r}}{M}] \mid
  \s{r}{\oc V}$ from which we derive
  $$
  \inference
  {R;\Gamma',x:(\delta + 1) \vdash^\delta M}
  {R;\Gamma' \vdash^\delta \letb{x}{\get{r}}{M}}
  $$
  for some $\delta \geq 0$ with $r:\delta \in R$, and some $\Gamma'$
  extending $\Gamma$. We also derive
  $$
  \inference
  {R;\Gamma \vdash^{\delta+1} V}
  {
    \inference
    {R;\Gamma \vdash^{\delta} \oc V}
    {R;\Gamma \vdash^0 \store{r}{\oc V}}
  }
  $$
  By (2) we get $R;\Gamma' \vdash^{\delta+1} V$.
  By (3) we derive $$R;\Gamma' \vdash^\delta \sub{M}{V/x}$$
  hence
  $$R;\Gamma \vdash^0 E[\sub{M}{V/x}]$$
  and finally
  $$R;\Gamma \vdash^0 E[\sub{M}{V/x}] \mid \store{r}{\oc V}$$

Concerning the depth bound, by (3), the depth of an occurrence in $
E[\sub{M}{V/x}] \mid \store{r}{\oc V}$ is bounded by the depth of an
occurrence which is already in $E[\letb{x}{\get{r}}{M}] \mid
\store{r}{\oc V}$, hence $d(E[\sub{M}{V/x}] \mid \store{r}{\oc V})
\leq d(E[\letb{x}{\get{r}}{M}] \mid \store{r}{\oc V})$.

\end{itemize}
\end{enumerate}

\subsection{Proof of proposition \ref{prop-termination}}
 We do this by case analysis on the
  reduction rules. 

\begin{itemize}
\item $P=E[(\lambda x.M)V] \reduc P' = E[\sub{M}{V/x}]$\\
Let the occurrence of the redex $(\lambda x.M)V$ be
at depth $i$.
The restrictions on the formation of terms require that
$x$ occurs at most once in $M$ at depth $0$.
Then $\occ_i(P) -3 \geq \occ_i(P')$ because we remove
the nodes for application and $\lambda$-abstraction and either 
$V$ disappears or the occurrence of the variable $x$ in $M$ 
disappears (both being at the same depth as the redex).
Clearly $\occ_j(P) = \occ_j(P')$ if $j \neq i$, hence 
\begin{equation}
\label{lambda-decrease}
\mu_n(P') \leq \\(\occ_n(P), \ldots, \occ_{i+1}(P),\occ_i(P) - 3,\mu_{i-1}(P))
\end{equation}
and
$\mu_n(P)>\mu_n(P')$.

\item $P=E[\letm{x}{!V}{M}] \reduc P' = E[\sub{M}{V/x}]$\\
Let the occurrence of the redex 
$\letm{x}{!V}{M}$ be at depth $i$.
The restrictions on the formation of terms require that
$x$ may only occur in $M$ at depth $1$ and hence in $P$
at depth $i+1$. 
We have that $\occ_i(P') = \occ_i(P) -2$ 
 because the $\letbang$ node disappear.
Clearly,  $\occ_j(P) = \occ_j(P')$ if $j<i$.
The number of occurrences of $x$ in $M$ is bounded by 
$k=\occ_{i+1}(P) \geq 2$. 
Thus if  $j>i$  then $\occ_{j}(P') \leq k \cdot
\occ_j(P)$. Let's write, for $ 0 \leq i \leq n$: 
$$\mu_n^i(P) \cdot k = (\occ_n(P) \cdot k,
\occ_{n-1}(P) \cdot k, \ldots, \occ_{i}(P) \cdot k)$$
Then we have
\begin{equation}
\label{letbang-decrease}
  \mu_n(P') \leq (\mu_n^{i+1}(P) \cdot k, \occ_i(P) -2, \mu_{i-1}(P))
\end{equation}
and finally 
$\mu_n(P)>\mu_n(P')$.

\item $P\equiv E[\st{r}{V}]  \arrow P' \equiv E[*] \mid
  \store{r}{V}$\\
If $R;\Gamma \vdash^{\delta}  \st{r}{V}$ then by~\ref{prop-depth-effect}(4) we have $R;\Gamma
\vdash^0 \store{r}{V}$ with $R(r) =
\delta$. Hence, by definition of the depth, the occurrences
in $V$ stay at depth $\delta$ in $\store{r}{V}$. Moreover, the node
$\st{r}{V}$ disappears and the nodes $*$, $\para$, and $r
\leftarrow$ appear. Recall that we assume the
occurrences $\para$ and $r \leftarrow$ do not count in the
measure and that $\mathsf{set}(r)$ counts for two occurrences.
Thus $\occ_{\delta}(P') = \occ_{\delta}(P) -
2 + 1 + 0 + 0$. The number of occurrences at other depths stay unchanged, hence
$\mu_n(P) > \mu_n(P')$.

\item $P \equiv E[\get{r}] \mid \store{r}{V} \arrow P'
  \equiv E[V]$ \\
  If $R;\Gamma \vdash^0 \store{r}{V}$ with $R(r) = \delta$, then
  $\get{r}$ must be at depth $\delta$ in $E[~]$. Hence, by definition
  of the depth, the occurrences in $V$ stay at depth $\delta$, while
  the node $\get{r}$ and $\para$
  disappear. Thus $\occ_{\delta}(P') = \occ_{\delta}(P) - 1$ and the
  number of occurrences at other depths stay unchanged, hence
  $\mu_n(P) > \mu_n(P')$.

\item $P \equiv E[\letm{x}{\get{r}}{M}] \mid \store{r}{\bang V} \arrow
P' \equiv E[\sub{M}{V/x}] \mid \store{r}{\bang V}$\\
This case is the only source of duplication with the reduction rule
on $\letbang$.
Suppose $R;\Gamma \vdash^{\delta} \letm{x}{\get{r}}{M}$. Then
we must have $R;\Gamma \vdash^{\delta+1} V$.
The restrictions on the formation of terms require that
$x$ may only occur in $M$ at depth $1$ and hence in $P$
at depth $\delta+1$. Hence the occurrences in $V$ stay at the same
depth in $\sub{M}{V/x}$, while
the $\mathsf{let}$, $\get{r}$ and some $x$ nodes disappear, hence
$\occ_{\delta}(P) \leq \occ_{\delta}(P') -2$. The number of occurrences of $x$ in $M$ is bounded by 
$k=\occ_{\delta+1}(P) \geq 2$. 
Thus if  $j>\delta$  then $\occ_{j}(P') \leq k \cdot \occ_j(P)$.
Clearly,  $\occ_j(M) = \occ_j(M')$ if $j<i$. Hence, we have
\begin{equation}
  \label{letget-decrease}
  \mu_n(P') \leq (\mu_n^{i+1}(P) \cdot k, \occ_i(P)-2, \mu_{i-1}(P))
\end{equation}
and $\mu_n(P) > \mu_n(P')$. 
\end{itemize}

\subsection{Proof of lemma \ref{shift-lemma}}
We start by remarking some basic inequalities.

\begin{lemma}[some inequalities]\label{inequalities-lemma}
The following properties hold on natural numbers.

\begin{enumerate}

\item $\qqs{x\geq 2, y\geq 0}{(y+1) \leq x^y}$

\item $\qqs{x\geq 2, y\geq 0}{(x\cdot y) \leq x^y}$

\item $\qqs{x\geq 2, y,z \geq 0}{(x\cdot y)^z \leq x^{(y\cdot z)}}$

\item $\qqs{x\geq 2, y\geq 0, z\geq 1}{x^z\cdot y \leq x^{(y\cdot z)}}$

\item If $x\geq y\geq 0$ then $(x-y)^k \leq (x^k - y^k)$

\end{enumerate}
\end{lemma}
\begin{proof}
\begin{enumerate}
\item By induction on $y$. The case for $y=0$ is clear.
For the inductive case, we notice:
\[
(y+1)+1 \leq 2^y + 2^y = 2^{y+1} \leq x^{y+1}~.
\]
\item By induction on $y$. The case $y=0$ is clear.
For the inductive case, we notice:
\[
\begin{array}{llll}
x \cdot (y+1) &\leq  &x\cdot  (x^y)  &\mbox{(by (1))} \\
       &=  &x^{(y+1)}
\end{array}
\]
\item
By induction on $z$. The case $z=0$ is clear.
For the inductive case, we notice:
\[
\begin{array}{llll}
(x\cdot y)^{z+1}  &=      &(x\cdot y)^{z} (x\cdot y)  \\
            &\leq   &x^{y \cdot z} (x\cdot y)     &\mbox{(by inductive hypothesis)} \\
            &\leq   &x^{y\cdot z} (x^y)    &\mbox{(by (2))} \\
            &=      &x^{y\cdot (z+1)}
\end{array}
\]
\item
From $z\geq 1$ we derive $y\leq y^z$. Then:
\[
\begin{array}{llll}
x^z \cdot y  &\leq &x^z \cdot y^z  \\
       &=    &(x\cdot y)^z   \\
       &\leq &x^{y\cdot z}   &\mbox{(by (3))}
\end{array}
\]
\item
By the binomial law, we have $x^k= ((x-y) + y)^k = (x-y)^k+y^k+ p$  
with $p \geq 0$. Thus $(x-y)^k = x^k-y^k- p$ which implies
$(x-y)^k \leq x^k-y^k$. 
\end{enumerate}
\end{proof}

We also need the following property.

\begin{lemma}[pre-shift]
\label{pre-shift}
Assuming $\alpha \geq 1$ and $\beta \geq 2$, the following property
holds for the tower functions with $x,\vc{x}$ ranging over numbers
greater or equal than $2$: $$\beta \cdot t_\alpha(x,\vc{x}) \leq
t_\alpha(\beta \cdot x,\vc{x})$$
\end{lemma}
\begin{proof}
This follows from:
\[
\beta \leq \beta^{2^{t_\alpha(\vc{x})}}
\]  
\end{proof}

Then we can derive the proof of the shift lemma as follows.

Let $k=t_\alpha(x',\vc{x}) \geq 2$. Then 
\[
\begin{array}{lll}
t_{\alpha}(\beta \cdot x, x', \vc{x}) = \beta\cdot (\alpha \cdot x)^{2^{k}} 
&\leq (\alpha \cdot x)^{\beta \cdot 2^{k}}   &\mbox{(by lemma \ref{inequalities-lemma}(3))} \\
&\leq (\alpha \cdot x)^{(\beta\cdot 2)^k} \\
&\leq (\alpha \cdot x)^{2^{(\beta \cdot k)}} &\mbox{(by lemma \ref{inequalities-lemma}(3))}
\end{array}
\]
and by lemma \ref{pre-shift}
$\beta \cdot t_\alpha(x',\vc{x}) \leq t_\alpha(\beta \cdot x',\vc{x})$.

Hence 
$$
(\alpha \cdot x)^{2^{(\beta\cdot k)}} \leq
(\alpha \cdot x)^{2^{t_\alpha(\beta \cdot x',\vc{x})}} =
t_\alpha(x,\beta \cdot x',\vc{x})
$$

\subsection{Proof of theorem \ref{elementary-bound-effect}}
Suppose $\mu_\alpha(P)=(x_0,\ldots,x_\alpha)$ so that $x_i$ corresponds
to the occurrences at depth $(\alpha-i)$ for $0\leq i \leq \alpha$.
Also assume the reduction is at depth $(\alpha - i)$. By looking at
equations~\eqref{lambda-decrease} and~\eqref{letbang-decrease} in the
proof of termination (Proposition~\ref{prop-termination}), we see that the components
$i+1,\ldots,\alpha$ of $\mu_\alpha(P)$ and $\mu_\alpha(P')$ coincide. 
Hence, let $k=2^{t_{\alpha}(x_{i+1},\ldots,x_{\alpha})}$.
By definition of the tower function, $k\geq 1$. 

We proceed by case analysis on the reduction rules.
\begin{itemize}
\item $P \equiv \letb{x}{\oc V}{M} \arrow P' \equiv \sub{M}{V/x}$\\
By inequality~\eqref{letbang-decrease} we know that: 
\[
\begin{array}{r c l}
t_\alpha(\mu_\alpha(P'))
&\leq& t_\alpha(x_0 \cdot x_{i-1},\ldots,x_{i-1}\cdot x_{i-1},x_{i}-2,x_{i+1},\ldots,x_{\alpha}) \\
&=& t_\alpha(x_0\cdot x_{i-1},\ldots, x_{i-1}\cdot x_{i-1},x_{i}-2)^{k}
\end{array}
\]
By iterating lemma~\ref{shift-lemma}, we derive:
\[
\begin{array}{r c l}
&&t_\alpha(x_0\cdot x_{i-1},x_{1}\cdot x_{i-1},\ldots, x_{i-1}\cdot x_{i-1},x_{i}-2) \\
&\leq& t_\alpha(x_0,x_{1} \cdot x_{i-1}^{2}, \ldots, x_{i-1} \cdot x_{i-1},x_{i}-2) \\
&\leq& \ldots \\
&\leq& t_{\alpha}(x_0,x_1,\ldots, x_{i-1}^{i},x_{i}-2)
\end{array}
\]
Renaming $x_{i-1}$ with $x$ and $x_{i}$ with $y$, we are left to show that:
\[
(\alpha x^{i})^{2^{(\alpha \cdot (y-2))^{k}}} 
<
(\alpha x)^{2^{(\alpha \cdot y)^{k}}} 
\]
Since $i\leq \alpha$ the first quantity is bounded by:
\[
(\alpha x)^{\alpha \cdot 2^{(\alpha \cdot (y-2))^{k}}}
\]

We notice:
\[
\begin{array}{rcl}
&&\alpha \cdot 2^{(\alpha \cdot (y-2))^{k}}  \\
&=& \alpha \cdot 2^{(\alpha\cdot y - \alpha \cdot 2)^{k}} \\
&\leq& \alpha \cdot 2^{(\alpha \cdot y)^{k} - (\alpha \cdot 2)^{k}}
\qquad \mbox{(by lemma \ref{inequalities-lemma}(5))}
\end{array}
\]
So we are left to show that:
$$\alpha 2^{(\alpha\cdot y)^k - (\alpha \cdot 2)^k)} \leq  2^{(\alpha \cdot y)^{k}}$$
Dividing by  $2^{(\alpha \cdot y)^{k}}$ and  recalling that $k\geq 1$, it remains to check:
\[
\alpha \cdot 2^{-(\alpha \cdot 2)^{k}} \leq  \alpha \cdot 2^{-(\alpha \cdot 2)} < 1
\]
which is obviously true for $\alpha \geq 1$.

\item $P \equiv (\lambda x.M)V \arrow P' \equiv  \sub{M}{V/x}$\\
By equation~\eqref{lambda-decrease}, we have that:
\[
t_\alpha(\mu_\alpha(P')) 
\leq t_\alpha(x_0 ,\ldots,x_{i-1},x_{i}-2,x_{i+1},\ldots,x_{\alpha}) 
\]
and one can check that this quantity is strictly less than:
\[
t_\alpha(\mu_\alpha(P)) =  t_\alpha(x_0 ,\ldots,x_{i-1},x_{i},x_{i+1},\ldots,x_{\alpha}) 
\]

\item $P \equiv \letb{x}{\get{r}}{M} \mid \store{r}{\oc V} \arrow
  P' \equiv \sub{M}{V/x} \para \store{r}{\oc V} $\\
 Let $k=2^{t_{\alpha}(x_{i+1},\ldots,x_{\alpha})}$.
By definition of the tower function, $k\geq 1$. 
By equation~\eqref{letget-decrease} we have
\[
\begin{array}{ll}
t_\alpha(\mu_\alpha(P')) 
&\leq t_\alpha(x_0 \cdot x_{i-1},\ldots,x_{i-1}\cdot x_{i-1},x_{i}-2,x_{i+1},\ldots,x_{\alpha}) \\
&= t_\alpha(x_0\cdot x_{i-1},\ldots, x_{i-1}\cdot x_{i-1},x_{i}-2)^{k}
\end{array}
\]
And we end up in the case of the rule for $\letbang$.

\item For the read that consume a value from the store, by looking at
  the proof of termination, we see that exactly one
  element of the vector $\mu_{\alpha}(P)$ is strictly decreasing during
  the reduction, hence one
  can check that $t_\alpha(\mu_\alpha(P)) > t_\alpha(\mu_\alpha(P'))$.

\item The case for the write is similar to the read.
\end{itemize}

We conclude with the following remark that shows that the size of a
program is proportional to its number of occurrences.
\begin{remark}
  The size of a program $\size{P}$ of depth $d$ is at most twice
  the sum of its occurrences: $\size{P} \leq 2 \cdot \sum_{0 \leq i
    \leq d} \occ_i(P)$.
\end{remark}
Hence the size of a program $P$ is bounded by $t_d(\mu_d(P))$.
\subsection{Proof of proposition \ref{type-substitution}}
 By induction on $A$.
  \begin{itemize}
  \item $A \equiv t'$\\
    We have
    $$
    \inference
    {R \vdash t' & t \notin R}
    {R \vdash \forall t.t'}
    $$
    If $t \neq t'$ we have $\sub{t'}{B/t} \equiv t'$ hence $R \vdash
    [B/t]t'$. If $t \equiv t'$ then we have $\sub{t'}{B/t} \equiv B$ hence
    $R \vdash \sub{t'}{B/t}$. 
  \item $A \equiv \tertype$\\
    We have 
    $$
    \inference
    {R \vdash \tertype & t \notin R}
    {R \vdash \forall t.\tertype}
    $$
    from which we deduce $R \vdash \sub{\tertype}{B/t}$.
  \item $A \equiv (C \multimap D)$\\
    By induction hypothesis we have $R \vdash \sub{C}{B/t}$ and $R \vdash
    \sub{D}{B/t}$. We then derive
    $$
    \inference
    {R \vdash \sub{C}{B/t} & R \vdash \sub{D}{B/t}}
    {R \vdash \sub{(C \multimap D)}{B/t}}
    $$
  \item $A \equiv \oc C$\\
    By induction hypothesis we have $R \vdash \sub{C}{B/t}$, from which we
    deduce
    $$
    \inference
    {R \vdash \sub{C}{B/t}}
    {R \vdash \sub{\oc C}{B/t}}
    $$
  \item $A \equiv \rgtype{r}{C}$\\
    We have 
    $$
    \inference
    {\inference
      {R \vdash & r:C \in R}
      {R \vdash \rgtype{r}{C}}
      & t \notin R}
    {R \vdash \forall t.\rgtype{r}{C}}
    $$
    As $t \notin R$ and $r:(\delta,C) \in R$, we have $r:(\delta,\sub{C}{B/t}) \in R$, from
    which we deduce
    $$
    \inference
    {R\vdash & r:(\delta,\sub{C}{B/t} \in R}
    {R \vdash \sub{\rgtype{r}{C}}{B/t}}
    $$
  \item $A \equiv \forall t'.C$\\
    If $t \neq t'$: From $R \vdash \forall t.(\forall t'.C)$ we have $t' \notin R$ and
    by induction hypothesis we have $R \vdash \sub{C}{B/t}$, from which we
    deduce
    $$
    \inference
    {R \vdash \sub{C}{B/t} & t' \notin R}
    {R \vdash \sub{(\forall t'.C)}{B/t}}
    $$
    If $t \equiv t'$ we have $\sub{(\forall t'.C)}{B/t} \equiv \forall t'.C$.
    Since we have 
    $$
    \inference
    {R \vdash \forall t'.C & t \notin R}
    {R \vdash \forall t.(\forall t'.C)}
    $$
    we conclude $R \vdash \sub{(\forall t'.C)}{B/t}$.
  \end{itemize}

\subsection{Proof of theorem \ref{progress-thm}}
Properties 1 and 2 are easily checked. 

\subsubsection{Substitution}
If $x$ is not free in $M$, we just have to check that
any proof of $\Gamma,x:(\delta',A) \vdash^\delta M$ can be transformed 
into a proof of $\Gamma \vdash^\delta M$. 

So let us assume $x$ is free in $M$.   
Next, we proceed by induction on the derivation of $\Gamma,x:\delta'\vdash^\delta M$.

\begin{itemize}

\item $\Gamma,x:(\delta,A) \vdash^\delta x:A$. Then $\delta=\delta'$, $\sub{x}{V/x}=V$,
and by hypothesis $\Gamma \vdash^{\delta} V : A$.

\item $\Gamma,x:(\delta',A) \vdash^{\delta} \lambda y.M: B \multimap C$ 
is derived from $\Gamma,x:(\delta',A),y:(\delta,B) \vdash^\delta M:C$,  with $x\neq y$ 
and $y$ not occurring in $V$.
By (2), $\Gamma,y:(\delta,B) \vdash^{\delta'} V:A$. 
By inductive hypothesis, $\Gamma,(y:\delta,B) \vdash^\delta
\sub{M}{V/x}: C$, 
and then we conclude $\Gamma \vdash^\delta \sub{(\lambda y.M)}{V/x}: B
\multimap C$.

\item $\Gamma,x:(\delta',A)\vdash^\delta (M_1M_2): C$ is derived from
  $\Gamma,x:(\delta',A)\vdash^\delta M_1 :B \multimap C$ and
  $\Gamma,x:(\delta',A)\vdash^\delta M_1 :B \multimap C$. By
  inductive hypothesis, $\Gamma \vdash^\delta \sub{M_1}{V/x}: B
  \multimap C$ and $\Gamma \vdash^\delta \sub{M_2}{V/x}:
   C$, and then we conclude $\Gamma \vdash^\delta
  \sub{(M_1M_2)}{V/x}: C$.

\item $\Gamma,x:(\delta',A) \vdash^\delta \bang M: \oc B$ is derived from
$\Gamma,x:(\delta',A) \vdash^{\delta+1} M:B$. 
By inductive hypothesis, $\Gamma \vdash^{\delta+1} \sub{M}{V/x} : B$, and then we conclude
$\Gamma \vdash^{\delta} \sub{\bang M}{V/x} : \oc B$.

\item $\Gamma,x:(\delta',A) \vdash^\delta \letm{y}{M_1}{M_2}:B$, with 
$x\neq y$ and $y$ not free in $V$ is derived from
$\Gamma,x:(\delta',A) \vdash^\delta M_1:C$ and $\Gamma,x:(\delta',A),y:(\delta+1,C) \vdash^\delta M_2:B$.
By inductive hypothesis, $\Gamma \vdash^\delta \sub{M_1}{V/x}:C$ 
$\Gamma,y:(\delta+1,C) \vdash^\delta \sub{M_2}{V/x}:B$, and then we conclude
$\Gamma \vdash^\delta \sub{(\letm{y}{M_1}{M_2})}{V/x}:B$.

 \item $M \equiv \get{r}$.
  We have $R,r:(\delta,B);\Gamma,x:(\delta',A) \vdash^\delta \get{r}:B$.
  Since $\sub{\get{r}}{V/x} = \get{r}$ and $x \notin \fv{\get{r}}$
  then $R,r:(\delta,B);\Gamma \vdash^\delta \sub{\get{r}}{V/x}:B$.

\item $M \equiv \st{r}{V'}$.
  We have
  $$
  \inference
  {R,r:(\delta,C);\Gamma,x:(\delta',A) \vdash^{\delta} V':C}
  {R,r:(\delta,C);\Gamma,x:(\delta',A)  \vdash^{\delta} \st{r}{V'}:\tertype}
  $$
  By induction hypothesis we get
  $$
  R,r:(\delta,C);\Gamma \vdash^{\delta} \sub{V'}{V/x}:C
  $$
  and hence we derive
  $$
  R,r:(\delta,C);\Gamma  \vdash^{\delta} \sub{(\st{r}{V'})}{V/x}:\tertype
  $$

\item $M \equiv (M_1 \mid M_2)$.
  We have 
  $$
  \inference
  {R;\Gamma,x:(\delta',A) \vdash^\delta M_i:C_i & i=1,2}
  {R;\Gamma,x:(\delta',A) \vdash^\delta (M_1\mid M_2):\behtype}
  $$
  
  By induction hypothesis we derive 
  $$
  R;\Gamma \vdash^\delta \sub{M_i}{V/x}:C_i
  $$
  and hence we derive 
  $$
  R;\Gamma \vdash^\delta \sub{(M_1\mid M_2)}{V/x}:\behtype
  $$
\end{itemize}

\subsubsection{Subject Reduction}
We first state and sketch the proof of $4$ lemmas.

\label{sub-red-proof}
\begin{lemma}[structural equivalence preserves typing] \label{sub-red-equ}
If $R;\Gamma\vdash^\delta P:\alpha$ and $P\equiv P'$ then $R;\Gamma\vdash^\delta P':\alpha$.
\end{lemma}
\begin{proof}
Recall that structural equivalence is the least equivalence
relation induced by the equations stated in 
Table \ref{semantics} and closed under static contexts.
Then we proceed by induction on the proof of structural equivalence.
This is is mainly a matter of reordering the pieces of the typing
proof of $P$ so as to obtain a typing proof of $P'$.
\end{proof}

\begin{lemma}[evaluation contexts and typing]  \label{eva-sub-lem}
Suppose that in the proof of $R;\Gamma \vdash^\delta E[M]:\alpha$ 
we prove $R;\Gamma' \vdash^{\delta'} M:\alpha'$. Then replacing $M$ with a 
$M'$ such that $R;\Gamma'\vdash^{\delta'} M':\alpha'$, we can still derive
$R;\Gamma \vdash^\delta E[M']:\alpha$.
\end{lemma}
\begin{proof}
By induction on the structure of $E$.
\end{proof}

\begin{lemma}[functional redexes]\label{fun-redex}
If $R;\Gamma \vdash^\delta E[\Delta] :\alpha$ where 
$\Delta$ has the shape $(\lambda x.M)V$ or $\letm{x}{\oc V}{M}$ then
$R;\Gamma \vdash^\delta E[\sub{M}{V/x}]:\alpha$.
\end{lemma}
\begin{proof}
We appeal to the substitution
lemma \ref{sub-lemma}. This settles the case
where the evaluation context $E$ is trivial. If it is complex
then we also need  lemma \ref{eva-sub-lem}.
\end{proof}

\begin{lemma}[side effects redexes]\label{side-eff-redex}
If $R;\Gamma \vdash^\delta \Delta:\alpha$ where
$\Delta$ is one of the programs on the left-hand
side then $R;\Gamma \vdash^\delta \Delta':\alpha$ where
$\Delta'$ is the corresponding program on the right-hand side:
\[
\begin{array}{lc|c}

(1) &E[\st{r}{V}]       &E[*]\mid \store{r}{V} \\
(2) &E[\get{r}] \mid \store{r}{V}  &E[V] \\
(3)\quad &E[\letb{x}{\get{r}}{M}] \mid \store{r}{\bang V}  \quad & \quad E[\sub{M}{V/x}] \mid \store{r}{\bang V}

\end{array}
\]
\end{lemma}
\begin{proof}
We proceed by case analysis.

\begin{enumerate}

\item Suppose we derive $R;\Gamma \vdash^\delta E[\st{r}{V}]:\alpha$ from
$R;\Gamma' \vdash^{\delta'} \st{r}{V}:\tertype$. We can derive
$R;\Gamma'\vdash^{\delta'} *:\tertype$ and by
Lemma~\ref{eva-sub-lem} we derive $R;\Gamma \vdash^\delta
E[*]:\alpha$ and finally $R;\Gamma \vdash^\delta E[\st{r}{V}] \mid \store{r}{V}:\alpha$.

\item Suppose $R;\Gamma \vdash^\delta E[\get{r}]:\alpha$
  is derived from $R;\Gamma \vdash^{\delta'} \get{r} :
  A$,
  where $r:(\delta',A) \in R$. Hence $R;\Gamma \vdash^0
  \store{r}{V}:\behtype$ is derived from $R;\Gamma \vdash^{\delta'} V:
  A$. Finally, by
  Lemma~\ref{eva-sub-lem} we derive $R;\Gamma \vdash^\delta E[V]:\alpha$.

\item Suppose $R;\Gamma \vdash^\delta E[\letb{x}{\get{r}}{M}]:\alpha$ is derived from
$$
\inference
{R;\Gamma' \vdash^{\delta'} \get{r}:\oc A & R;\Gamma',x:(\delta'+1,A)
  \vdash^{\delta'} M : \alpha'}
{R;\Gamma' \vdash^{\delta'} \letb{x}{\get{r}}{M}:\alpha'}
$$
where $r:(\delta',\oc A) \in R$. Hence
$R;\Gamma \vdash^0 \store{r}{\oc V}:\behtype$ is derived from
$R;\Gamma \vdash^{\delta'+1}  V:  A$. By Lemma~\ref{sub-lemma} we can
derive $R;\Gamma' \vdash^{\delta'} \sub{M}{V/x}:\alpha'$. Then by
Lemma~\ref{eva-sub-lem} we derive $R;\Gamma \vdash^\delta E[\sub{M}{V/x}]:\alpha$.
\end{enumerate}
\end{proof}

We are then ready to prove subject reduction.
We recall that $P\arrow P'$ means that 
$P$ is structurally equivalent to a program
$C[\Delta]$ where $C$ is a static context, $\Delta$ 
is one of the programs on the left-hand side of the
rewriting rules specified in Table \ref{semantics},
$\Delta'$ is the respective program on the right-hand side,
and $P'$ is syntactically equal to $C[\Delta']$.

By lemma \ref{sub-red-equ}, we know that 
$R;\Gamma \vdash^\delta C[\Delta]:\alpha$.
This entails that $R';\Gamma' \vdash^{\delta'} \Delta : \alpha'$ for
suitable $R',\Gamma',\alpha',\delta'$.
By lemmas \ref{fun-redex} and \ref{side-eff-redex}, we derive that
$R';\Gamma'\vdash^{\delta'} \Delta':\alpha'$.
Then by induction on the structure of $C$ we argue 
that $R;\Gamma \vdash^\delta C[\Delta']:\alpha$. 

\subsubsection{Progress}
To derive the progress property we first determine for each closed
type $A$ where $A = A_1 \multimap A_2$ or $A = \oc A_1$
the shape of  a closed value of type $A$ with the following classification lemma.
\begin{lemma}[classification]
  \label{lem-class}
  Assume $R;- \vdash^\delta V : A$.  Then:
  \begin{itemize}
  \item if $A = A_1 \multimap A_2$ then $V = \lambda x.M$,
  \item if $A = \oc A_1$ then $V = \oc V_1$
  \end{itemize}
\end{lemma}
\begin{proof}
  By case analysis on the typing rules.
  \begin{itemize}
  \item if $A = A_1 \multimap A_2$, the only typing rule that can be applied is
    $$
    \inference
    {R;x:(\delta,A_1) \vdash^\delta  M : A_2}
    {R;- \vdash^\delta \lambda x.M : A_1 \multimap A_2}
    $$
    hence $V = \lambda x.M$.
  \item if $A = \oc A_1$, the only typing rule that can be applied is
    $$
    \inference
    {R;- \vdash^{\delta+1}  V_1 : A_1}
    {R;- \vdash^\delta \oc V_1 : \oc A_1}
    $$
    hence $V = \oc V_1$.
  \end{itemize}
\end{proof}
\hfill\\
Then we proceed by induction on the structure of the threads $M_i$ to
show that each one of them is either a value or a stuck get of the
form $E[\Delta$] where $\Delta$ can be $(\lambda x.M)\get{r}$ or $\letb{x}{\get{r}}{M}$.
\begin{itemize}
\item $M_i = x$\\
  the case of variables is void since they are not closed terms.
\item $M_i = *$ or $M_i = r$ or $M_i = \lambda x.M$\\
  these cases are trivial since $*$, $r$ and $\lambda x.M$ are already values.
\item $M_i = PQ$\\
  We know that $PQ$ cannot reduce, which by looking at the evaluation
  contexts means that $P$ cannot reduce. Then by induction hypothesis
  we have two cases: either $P$ is a value or $P$ is a stuck get. 
  \begin{itemize}
  \item assume $P$ is a value. We have 
  $$
  \inference
  {R;- \vdash^\delta P :A \multimap B & R;- \vdash^\delta Q: A}
  {R;- \vdash^\delta PQ : B}
  $$ 
  By Lemma~\ref{lem-class} we have $P = \lambda x.M$. 
  Since $PQ$ cannot reduce and $P = \lambda x.M$, by looking at the
  evaluation contexts we have that $Q$ cannot reduce. Moreover $Q$
  cannot be a value, otherwise $PQ$ is a redex. Hence by induction
  hypothesis $Q$ is a stuck get of the form $E_1[\Delta]$. Hence
  $PQ$ is of the form $E[\Delta]$ where $E = PE_1$.
\item assume $P$ is a stuck get of the form $E_1[\Delta]$. Then $PQ$
  is of the form $E[\Delta]$ where $E = E_1Q$.
  \end{itemize}
\item $M_i = \letb{x}{P}{Q}$\\
  We know that $\letb{x}{P}{Q}$ cannot reduce, which by looking at the
  evaluation contexts means that $P$ cannot reduce. Then by induction hypothesis
  we have two cases: either $P$ is a value or $P$ is a stuck get.
  \begin{itemize}
  \item assume $P$ is a value. We have 
  $$
  \inference
  {R;- \vdash^\delta P : \oc A & R;x:(\delta+1,A) \vdash^\delta Q: B}
  {R;- \vdash^\delta \letb{x}{P}{Q} : B}
  $$ 
  By Lemma~\ref{lem-class} we have $P = \oc V$ hence $\letb{x}{\oc V}{Q}$
  is a redex and this contradicts the hypothesis that $\letb{x}{P}{Q}$
  cannot reduce. Thus $P$ cannot be a value.
\item assume $P$ is a stuck get of the form $E_1[\Delta]$. Then $\letb{x}{P}{Q}$
  is of the form $E[\Delta]$ where $E = \letb{x}{E_1}{Q}$.
  \end{itemize}
\item $M_i = \oc P$\\
  We know that $\oc P$ cannot reduce, which by looking at the
  evaluation contexts means that $P$ cannot reduce. Then by induction hypothesis
  we have two cases: either $P$ is a value or $P$ is a stuck get.
  \begin{itemize}
  \item assume $P$ is a value. Then $\oc P$ is also a value and we are
    done.
  \item assume $P$ is of the form $E_1[\Delta]$. Then $\oc P$ is of
    the shape $E[\Delta]$ where $E = \oc E_1$.
  \end{itemize}
\item $M_i = \get{r'}$\\
  We know that $\get{r'}$ cannot reduce which means that $M_i$ is of
  the form $E[\Delta]$ where $r'=r$ and $E=[]$ and that no value is
  associated with $r$ in the store.
\item $M_i = \st{r}{V}$\\
  This case is void since $\st{r}{V}$ can reduce is any case.
\end{itemize}



\subsection{Proof of theorem~\ref{complete-thm}}
\label{asec-completeness}
Elementary functions are characterized as the smallest class of
functions containing zero, successor, projection, subtraction and
which is closed by composition and bounded summation/product.
We will need the arithmetic functions defined in
Table~\ref{church-encodings-full}.
\begin{table}[h]
\[
\begin{array}{r c l l}
\nat &=& \forall t.\bang (t\limp t) \limp \bang (t\limp t)  &\mbox{(type of numerals)}\\\\

\zero &:& \nat                   &\mbox{(zero)}\\  
\zero &=& \lambda f.\bang(\lambda x.x)  \\\\

\succe &:& \nat \limp \nat     &\mbox{(successor)}\\ 
\succe &=& \lambda n.\lambda f.\s{let} \ \bang f=f \ \s{in} \\
  && \s{let} \ \bang y = n\bang f \ \s{in}  
  \bang (\lambda x.f(yx)) \\ \\ 

\churchn{n}&:& \nat         &\mbox{(numerals)}\\
\churchn{n}&=& \lambda f.\letm{f}{f}{\bang(\lambda x.f(\cdots (fx) \cdots))}\\\\

\add &:& \nat\limp (\nat\limp \nat)  &\mbox{(addition)}\\
\add &=& \lambda n.\lambda m. \lambda f. \s{let} \ \bang f = f \ \s{in}\\
&& \s{let} \ \bang y =  n\bang f \ \s{in} \\
  && \s{let} \ \bang y' = m \bang f \ \s{in} 
   \quad \bang (\lambda x.y(y'x)) \\ \\ 

\mult &:& \nat \limp (\nat\limp \nat) &\mbox{(multiplication)}\\ 
\mult&=& \lambda n.\lambda m. \lambda f. \s{let} \ \bang f = f \ \s{in}\\
  && n(m \bang f) \\ \\ 

\iit &:& \nat \limp \forall t.\bang(t\limp t) \limp \bang t  \limp \bang t   &\mbox{(iteration)} \\
\iit&= &\lambda n.\lambda g.\lambda x. \s{let} \ \bang y = ng \ \s{in}\\
 && \s{let} \ \bang y' = x \ \s{in} \ \bang (y y')   \\ \\ 

  \git &:& \forall t.\forall t'.\oc(t\multimap
  t) \multimap (\oc(t \multimap t) \multimap t') \multimap \nat
  \multimap t' \\
  \git &=& \lambda s.\lambda e.\lambda n. e(nts)\\\\


\end{array}
\]
\caption{Representation of some arithmetic functions}\label{church-encodings-full}
\end{table}
We will abbreviate $\lambdab{x}{M}$  for $\lambda x.\letb{x}{x}{M}$.
Moreover, in order to represent some functions, we need to manipulate
pairs in the language. We define the representation of pairs in
Table~\ref{rep-pairs}. 
\begin{table}[h]
$$
\begin{array}{rcll}
  A \times B &=& \forall t. (A \multimap B \multimap t) \multimap t &
  \mbox{(type of pairs)}\\\\
  \pair{M}{N} &:& A \times B & \mbox{(pair representation)}\\
  \pair{M}{N} &=& \lambda x.xMN\\\\

  \fst &:& \forall t,t'. t \times t' \multimap t & \mbox{(left destructor)}\\
  \fst &=& \lambda p.p(\lambda x.\lambda y.x)\\\\

  \snd &:& \forall t,t'. t \times t' \multimap t'& \mbox{(right destructor)}\\
  \snd &=& \lambda p.p(\lambda x.\lambda y.y)
\end{array}
$$
\caption{Representation of pairs}
\label{rep-pairs}
\end{table}
In the
following, we show that the required functions can be represented in the
sense of Definition~\ref{def-function-rep} by adapting the proofs from Danos
and Joinet~\cite{Danos}.

\subsubsection{Successor, addition and multiplication}
We check that $\succe$ represents the \emph{successor} function $s$:
$$
\begin{array}{l}
  s : \mathbb{N} \mapsto \mathbb{N}\\
  s(x) = x + 1
\end{array}
$$
\begin{proposition}
  $\succe \Vdash s$.  
\end{proposition}
\begin{proof}
  Take $\emptyset \vdash^\delta M:\nat$ and $M \Vdash n$. We have $\emptyset
  \vdash^\delta \succe : \nat \multimap \nat$. We can show that $\succe \: M
  \treduc \churchn{s(n)}$, hence $\succe \: M \Vdash s(n)$. Thus
  $\succe \Vdash s$.
\end{proof}

We check that $\add$ represents the \emph{addition} function $a$:
$$
\begin{array}{l}
  a : \mathbb{N}^2 \mapsto \mathbb{N}\\
  a(x,y) = x + y
\end{array}
$$
\begin{proposition}
  $\add \Vdash a$.  
\end{proposition}
\begin{proof}
  For $i=1,2$ take $\emptyset \vdash^\delta M_i:\nat$ and $M_i \Vdash
  n_i$. We have $ \emptyset \vdash^\delta \add : \nat \multimap \nat \multimap
  \nat $.  We can show that $\add \: M_1M_2 \treduc
  \churchn{a(n_1,n_2)}$, hence $\add \: M_1M_2 \Vdash a(n_1,n_2)$. Thus $A
  \Vdash a$.
\end{proof}

We check that $\mult$ represents the multiplication function $m$:
$$
\begin{array}{l}
  m : \mathbb{N}^2 \mapsto \mathbb{N}\\
  m(x,y) = x * y
\end{array}
$$

\begin{proposition}
  $\mult \Vdash m$.  
\end{proposition}
\begin{proof}
  For $i=1,2$ take $\emptyset \vdash^\delta M_i:\nat$ and $M_i \Vdash
  n_i$. 
We have
$
\emptyset \vdash^\delta \mult : \nat \multimap \nat \multimap \nat
$. We can show that $\mult \: M_1M_2 \treduc
   \churchn{m(n_1,n_2)}$, hence $\mult \: M_1M_2 \Vdash m(n_1,n_2)$. Thus $\mult
  \Vdash m$.
\end{proof}

\subsubsection{Iteration schemes}
We check that $\iit$ represents the following iteration function $it$:
$$
\begin{array}{l}
  it : (\mathbb{N} \mapsto \mathbb{N}) \mapsto \mathbb{N}  \mapsto \mathbb{N} \mapsto \mathbb{N}\\
  it(f,n,x) = f^n(x)
\end{array}
$$
\begin{proposition}
\label{prop-it} $\iit \Vdash it$.
\end{proposition}
\begin{proof}
  We have $ \emptyset\vdash^\delta
\iit : \nat \multimap \forall t.\oc(t\multimap t)\multimap \oc t
\multimap \oc t $. Given $\emptyset \vdash^\delta M:\nat$ with $M \Vdash n$, $\emptyset \vdash^\delta F:\nat\multimap\nat$ with $F \Vdash f$ and
 $\emptyset \vdash^\delta X:\nat$ with $X \Vdash x$, we observe that
  $\iit \: M (\oc F) (\oc X) \treduc F^nX$.  Since $F
  \Vdash f$ and $X \Vdash x$, we get $F^nX \treduc
   \churchn{it(f,n,x)}$. Hence $\iit \Vdash it$.
\end{proof}
The function $it$ is an instance of the more general iteration scheme $git$:
$$
\begin{array}{l}
  git : (\mathbb{N} \mapsto \mathbb{N}) \mapsto ((\mathbb{N} \mapsto \mathbb{N}) \mapsto \mathbb{N})  \mapsto \mathbb{N} \mapsto \mathbb{N}\\
  git(step,exit,n) = exit(\lambda x.step^n(x))
\end{array}
$$
Indeed, we have:
$$git(f,\lambda f.fx,n) = (\lambda f.fx)(\lambda x.f^n(x)) = it(f,n,x)$$
\begin{proposition}
  $\git \Vdash git$.
\end{proposition}
\begin{proof}
  Take $\emptyset \vdash^\delta M : \nat$ with $M \Vdash n$,
  $\emptyset \vdash^\delta E : ((\nat \multimap \nat) \multimap \nat)
  \multimap \nat$ with $E \Vdash exit$, $\emptyset \vdash^\delta S :
  \nat \multimap \nat$ with $S \Vdash step$. Then we have $\git \: S
  \: E \: M \treduc E (\lambda x.S^nx)$. Since $S \Vdash step$ and $E
  \Vdash exit$ we have $E (\lambda x.S^nx) \treduc
  \churchn{exit(\lambda x.step^n(x)}$. Hence $\git \Vdash git$.
\end{proof}

\subsubsection{Coercion}
Let $S = \lambdal{n}{N}.S'$. For $0 \geq i$, we define $S'_i$ inductively:
$$
\begin{array}{l}
  S'_0 =  S'\\
  S'_i = \letb{n}{n}{\oc S_{i-1}'}\\
\end{array}
$$
Let $S_i = \lambda n.S'_i$. We can derive $\emptyset
\vdash^\delta S_i
 : \oc^i\nat \multimap \oc^i\nat$.
For $i \geq 0$, we define $C_i$ inductively:
$$
\begin{array}{l}
C_0 = \lambda x.x\\
C_{i+1} = \lambda n.\iit  (\oc
S_{i}) (\oc^{i+1}  \churchn{0}) n\\  
\emptyset\vdash^\delta C_i
: \nat \multimap \oc^i \nat
\end{array}
$$

\begin{lemma}[integer representation is preserved by coercion]
\label{lem:int-rep-coe} 
  Let $\emptyset \vdash^\delta M:\nat$ and $M \Vdash n$. We can derive
  $\emptyset \vdash^\delta C_iM:\oc^i \nat$. Moreover $C_iM \Vdash n$.
\end{lemma}
\begin{proof}
  By induction on $i$.
\end{proof}

\begin{lemma}[function representation is preserved by coercion]
\label{lem:coerc}
Let $$\emptyset\vdash^\delta F:\oc^{i_1}\nat_1 \multimap \ldots \multimap
\oc^{i_k}\nat_k \multimap \oc^{p}\nat$$ and 
$\emptyset \vdash^\delta M_j:\nat$ with $M_j \Vdash n_j$ for $1 \leq j \leq
k$ such that
$F(\oc^{i_1}M_1\ldots(\oc^{i_k}M_k)) \treduc
 \churchn{f(n_1,\ldots,n_k)}$. Then we can find a term $\coerc{F} =
\lambdal{\vec{x}}{\nat}.F((C_{i_1}x_1)\ldots(C_{i_k}x_k))$ such that 
  $$\emptyset\vdash^\delta \coerc{F}:\nat \multimap \nat
\multimap \ldots \multimap \nat \multimap \oc^{p}\nat$$ and 
$\coerc{F} \Vdash f$.
\end{lemma}

\subsubsection{Predecessor and subtraction}
We first want to represent \emph{predecessor}:
$$
\begin{array}{l}
p : \mathbb{N} \mapsto \mathbb{N}\\
p(0) = 0\\
p(x) = x-1
\end{array}
$$

We define the following terms:
$$
\begin{array}{l}
  ST = \oc(\lambda z.\pair{\snd{z}}{f(\snd{z})})\\
f:(\delta+1,t\multimap t) \vdash^\delta ST : \oc(t\times t \multimap
t\times t)\\\\

EX = \lambda g.\letb{g}{g}{\oc(\lambda x.\fst{g\pair{x}{x}})}\\
\emptyset \vdash^\delta EX : \oc(t\times t\multimap t\times
t)\multimap \oc(t\multimap t)\\\\

P = \lambda n.\lambda f.\letb{f}{f}{\git \: ST \: EX \: n}\\
\emptyset \vdash^\delta P : \nat \multimap \nat
\end{array}
$$
\begin{proposition}[predecessor is representable]
  $P \Vdash p$.
\end{proposition}
\begin{proof}
  Take $\emptyset \vdash^\delta M:\nat$ and $M \Vdash
  n$. We can show that $\ftrad{PM} \treduc
   \churchn{p(n)}$, hence $PM \Vdash p(n)$. Thus $P
  \Vdash p$.
\end{proof}

Now we want to represent (positive) subtraction $s$:
$$
\begin{array}{l}
  s : \mathbb{N}^2 \mapsto \mathbb{N}\\
  s(x,y) = \left\lbrace 
    \begin{array}{l l}
      x-y & \textrm{if } x \geq y\\
      0 & \textrm{if } y \geq x
    \end{array}
  \right.
\end{array}
$$
Take
$$
\begin{array}{l}
SUB =
\lambda m.\letb{m}{m}{\lambda n .\iit \: \oc P \: \oc m \: n}:\oc\nat\multimap\nat\multimap\oc\nat\\
\emptyset \vdash^\delta SUB : \oc\nat \multimap
\nat\multimap\oc\nat\\  
\end{array}
$$
\begin{proposition}[subtraction is representable]
  $\coerc{SUB} \Vdash s$.
\end{proposition}
\begin{proof}
  For $i=1,2$ take $\emptyset \vdash^\delta M_i:\nat$ and $M_i \Vdash
  n_i$. We can show that $\ftrad{SUB (\oc M_1)M_2} \treduc
   \churchn{s(n_1,n_2)}$. Hence by Lemma \ref{lem:coerc},
  $\coerc{SUB} \Vdash s$.
\end{proof}


\subsubsection{Composition}
  Let $g$ be a $m$-ary function and $G$ be a term such that
  $\emptyset \vdash^\delta G:\nat_1 \multimap \ldots
  \multimap \nat_m\multimap \oc^p\nat$ (where $p
  \geq 0)$ and $G \Vdash g$.
  For $1 \leq i \leq m$, let $f_i$ be a $k$-ary function and $F_i$ a
  term such that $\emptyset \vdash^\delta
  F_i: \nat_1\multimap\ldots\multimap\nat_k\oc^{q_i}\nat$ (where $q_i \geq 0)$ and $F_i \Vdash f_i$.
  We want to represent the composition function $h$ such that:
  $$
  \begin{array}{l}
    h : \mathbb{N}^k \mapsto \mathbb{N}\\
    h(x_1,\ldots,x_k) = g(f_1(x_1,\ldots,x_k),\ldots,f_m(x_1,\ldots,x_k))
  \end{array}
  $$


  For $i \geq 0$ and a term $T$, we define $T^i$ inductively as:
$$
\begin{array}{l}
  T^0 = T\\
  T^i = \lambdal{\vec{x}}{\oc^i\nat}.\letb{\vec{x}}{\vec{x}}{\oc
    (T^{i-1}\vec{x})}
\end{array}
$$
  Let $q = max(q_i)$. We can derive $$\emptyset \vdash^\delta G^{q+1}:
  \oc^{q+1}\nat_1 \multimap \ldots \multimap \oc^{q+1}\nat_m \multimap \oc^{p+q+1}\nat$$ We can also
  derive $$\emptyset
  \vdash^\delta F^{q-q_i}_i :
  \oc^{q-q_i}\nat_1\multimap\ldots\multimap\oc^{q-q_i}\nat_k\multimap\oc^{q}\nat$$ 
  Then, applying coercion we get 
  $$\emptyset \vdash^\delta
  \coerc{F^{q-q_i}_i} : \nat_1\multimap\ldots\nat_k\multimap
  \oc^{q}\nat$$ and we derive $$x_1:(\delta+1,\nat),\ldots,x_k:(\delta
  +1,\nat) \vdash^\delta
  \oc (\coerc{F^{q-q_i}_i}x_1\ldots x_k) :
  \oc^{q+1}\nat$$
  Let $F'_i \equiv  \oc (\coerc{F^{q-q_i}_i}x_1\ldots x_k)$. 
  By application we get
  $$x_1:(\delta+1,\nat),\ldots,x_k:(\delta+1,\nat) \vdash^\delta
  G^{q+1}F'_1\ldots F'_m
  : \oc^{p+q+1}\nat$$ 
  We derive
  $$\emptyset \vdash^\delta
  \lambda \vec{x}.\letb{\vec{x}}{\vec{x}}{G^{q+1}F'_1\ldots F'_m}
  : \oc\nat_1\multimap\ldots\multimap\oc\nat_m\multimap\oc^{p+q+1}\nat$$ 
  Applying coercion we get
  $$\emptyset \vdash^\delta
  \coerc{\lambdal{\vec{x}}{\oc\nat}.\letb{\vec{x}}{\vec{x}}{G^{q+1}F'_1\ldots F'_m}}
  : \nat_1\multimap\ldots\multimap\nat_m\multimap\oc^{p+q+1}\nat$$ 
  Take $$H = \coerc{\lambdal{\vec{x}}{\oc\nat}.\letb{\vec{x}}{\vec{x}}{G^{q+1}F'_1\ldots F'_m}}$$

  \begin{proposition}[composition is representable]
    $H \Vdash h$.
  \end{proposition}
  \begin{proof}
    We now have to show that for
  all $M_i$ and $n_i$ where $1 \leq i \leq k$ such that $M_i \Vdash
  n_i$ and $\emptyset \vdash^\delta M_i:\nat$, we have
  $HM_1\ldots M_k \Vdash h(n_1,\ldots,n_k)$.
  Since $F_i
  \Vdash f_i$, we have $F_iM_1\ldots M_k \Vdash f_i(n_1,\ldots,n_k)$. Moreover $G
  \Vdash g$, hence $$G(F_1M_1\ldots M_k)\ldots(F_mM_1\ldots M_k) \Vdash
  g(f_1(n_1,\ldots,n_k),\ldots,f_m(n_1,\ldots,n_k))$$
  We can show that $HM_1\ldots M_k \treduc 
  G(F_1M_1\ldots M_k)\ldots(F_mM_1\ldots M_k)$, hence
  $$HM_1\ldots M_k \Vdash
  g(f_1(n_1,\ldots,n_k),\ldots,f_m(n_1,\ldots,n_k))$$
  Thus $H \Vdash h$.
  \end{proof}

\subsubsection{Bounded sums and products}
Let $f$ be a $k+1$-ary function $f:\mathbb{N}^{k+1} \rightarrow
\mathbb{N}$, where $$\emptyset \vdash
F:\nat_i\multimap\nat_1\multimap\ldots\multimap\nat_k\multimap\oc^p\nat$$
with $p \geq 0$ and $F \Vdash f$. We want to represent
\begin{itemize}
\item bounded sum: $\sum_{1 \leq i \leq n} f(i,x_1,\ldots,x_k)$
\item bounded product: $\prod_{1 \leq i \leq n} f(i,x_1,\ldots,x_k)$
\end{itemize}
For this we are going to represent $h: \mathbb{N}^{k+1} \rightarrow \mathbb{N}$:
$$
\begin{array}{l}
h(0,x_1,\ldots,x_k) = f(0,x_1,\ldots,x_k)\\
h(n+1,x_1,\ldots,x_k) = g(f(n+1,x_1,\ldots,x_k),h(n,x_1,\ldots,x_k))
\end{array}
$$
where $g$ is a binary function standing for addition or multiplication,
thus representable.
More precisely we have $g:\mathbb{N}^2 \rightarrow \mathbb{N}$
such that $\emptyset \vdash^\delta G : \nat \multimap \nat
\multimap \nat$ and $G \Vdash g$.

For $i \geq 0$ and a term $T$ 
we define $T^i$
inductively:
$$
\begin{array}{l}
  T^0 = Tx_1\ldots x_k\\
  T^i = \letb{x_1}{x_1}{\ldots\letb{x_k}{x_k}{\oc T^{i-1}}}
\end{array}
$$
We define the following terms:
$$
\begin{array}{l}
  ST = \lambda
z.\pair{S(\fst{z})}{G^p(Fx_1\ldots x_k(S(\fst{z})))(\snd{z})}\\
\emptyset;x_1:(\delta,\nat),\ldots,x_k:(\delta,\nat) \vdash^\delta ST : \nat
\times \oc^p\nat \multimap \nat \times \oc^p\nat\\\\

EX =
\lambda h.\letb{h}{h}{\oc\snd{h\pair{ \churchn{0}}{Fx_1\ldots
      x_k \churchn{0}}}}\\ 
\emptyset;x_1:(\delta+1,\nat),\ldots,x_k:(\delta+1,\nat) \vdash^\delta EX :
\oc(\nat\times\oc^p\nat\multimap\nat\times\oc^p\nat)\multimap\oc^{p+1}\nat

\end{array}
$$
We derive
$$
n:(\,\nat),\vec{x}:(\delta,\nat)\vdash^\delta \letb{\vec{x}}{\vec{x}}{\letb{n}{n}{\git \: \oc ST
\: EX \: n}} : \oc^{p+1}\nat
$$

Let $R = \letb{\vec{x}}{\vec{x}}{\letb{n}{n}{\git \: \oc ST
\: EX \: n}}$.
By coercion and abstractions we get
$$
\emptyset \vdash^\delta
\coerc{\lambda n.\lambda \vec{x}.R}
: \nat_i \multimap \nat_1\multimap\ldots\multimap\nat_k\multimap \oc^{p+1}\nat
$$
Take $H =
\coerc{\lambda n.\lambda \vec{x}.R}$.

\begin{proposition}[bounded sum/product is representable]
  $H\Vdash h$.
\end{proposition}
\begin{proof}
Given $M_i \Vdash i$ and
$M_j \Vdash n_j$ with $1 \leq j \leq k$ and taking $G$ for addition, we remark that
$$HM_iM_1\ldots M_k \treduc
 \churchn{f(i,n_1,\ldots,n_k) + \ldots + f(1,n_1,\ldots,n_k) +
  f(0,n_1,\ldots,n_k)}$$
Hence $H \Vdash h$.  
\end{proof}

\end{document}